\documentclass{article}
\usepackage{arxiv2c}
\usepackage{cite}
\usepackage{amsmath,amssymb,amsfonts}
\usepackage{algorithmic}
\usepackage{graphicx}
\usepackage{textcomp}
\usepackage{xcolor}
\usepackage[hyphens]{url}

\usepackage{amsmath,amsthm,amssymb,mathtools}
\usepackage{color}

\usepackage[T1]{fontenc}
\usepackage{listings}
\usepackage{multirow}
\usepackage{multicol}
\usepackage{url}
\usepackage{subfig}
\usepackage[ruled,vlined,linesnumbered]{algorithm2e}
\usepackage{tikz}
\usetikzlibrary{positioning}
\usetikzlibrary{calc}
\usepackage{ctable} %
\newtheorem{prop}{Proposition}
\usepackage{pgfplots}

\def\mytitle{Mapping Surface Code to Superconducting Quantum Processors}
\newcommand{\myCompilerName}{Surfmap}
\newcommand{\myCompilerNameSpace}{Surfmap }

\def\BibTeX{{\rm B\kern-.05em{\sc i\kern-.025em b}\kern-.08em
    T\kern-.1667em\lower.7ex\hbox{E}\kern-.125emX}}

\title{\mytitle}
\author{%
	Anbang Wu \\
	Department of Computer Science\\
	University of California, Santa Barbara \\
	\texttt{anbang@ucsb.edu} \\
	\And
	Gushu Li \\
	Department of Electrical \& Computer Engineering\\
	University of California, Santa Barbara \\
	\texttt{gushuli@ece.ucsb.edu} \\
	\AND
	Hezi Zhang \\
	Department of Computer Science \\
	University of California, Santa Barbara \\
	\texttt{hezi@ucsb.edu} \\
	\AND
	Gian Giacomo Guerreschi \\
	Intel Labs \\
	Santa Clara, California \\
	\texttt{gian.giacomo.guerreschi@intel.com}
	\AND
	Yufei Ding \\
	Department of Computer Science\\
	University of California, Santa Barbara \\
	\texttt{yufeiding@cs.ucsb.edu} \\
	\AND
	Yuan Xie \\
	Department of Electrical \& Computer Engineering\\
	University of California, Santa Barbara \\
	\texttt{yuanxie@ucsb.edu} \\
} 

\begin{document}
\maketitle
\thispagestyle{plain}
\pagestyle{plain}

\begin{abstract}

In this paper, we formally describe the three challenges of mapping surface code on superconducting devices, and present a comprehensive synthesis framework to overcome these challenges. The proposed framework consists of three optimizations. First, we adopt a geometrical method to allocate data qubits which ensures the existence of shallow syndrome extraction circuit. The proposed data qubit layout optimization reduces the overhead of syndrome extraction and serves as a good initial point for following optimizations. Second, we only use bridge qubits enclosed by data qubits and reduce the number of bridge qubits by merging short path between data qubits. The proposed bridge qubit optimization reduces the probability of bridge qubit conflicts and further minimizes the syndrome extraction overhead.  Third, we propose an efficient heuristic to schedule syndrome extractions. Based on the proposed data qubit allocation, we devise a good initial schedule of syndrome extractions and further refine this schedule to minimize the total time needed by a complete surface code error detection cycle. Our experiments on mainsstream superconducting quantum architectures have demonstrated the efficiency of the proposed framework.

\end{abstract}
\twocolumn
\section{Introduction}

Quantum hardware has advanced significantly in the last decade and demonstrated `Quantum Supremacy' for the first time in 2020~\cite{Arute2019QuantumSU}. 
Among various quantum hardware technologies~\cite{Kok2007LinearOQ, DiVincenzo2000ThePI, AsselmeyerMaluga20213DTQ, Google72Q},
the superconducting (SC) qubit is currently one of the most promising technique candidates for building quantum processors~\cite{Paik2011ObservationOH, Chen2014QubitAW} due to its low error rates, individual qubit addressability, fabrication scalability, etc. Many latest quantum processors adopt the SC technology, e.g., IBM's 65-qubit heavy-hexagon-architecture  chip~\cite{zhang2020high}, Rigetti's 32-qubit octagonal-architecture chip~\cite{gold2021experimental}, Google's 54-qubit 2D-lattice chip~\cite{Arute2019QuantumSU}. %

The low error rate %
of SC processors makes them an ideal platform for %
quantum error correction (QEC)~\cite{shor1995scheme, steane1996error, calderbank1996good, steane1996multiple, bravyi1998quantum} and thus realizing fault-tolerant (FT) quantum computation. Among various QEC codes, surface code~\cite{fowler2012surface} is a popular choice as it possesses one of the best error correction capabilities and can tolerate a high (up to 1\%) physical error rate. Even on noisy near-term devices, the surface code family can encode an arbitrarily accurate logical qubit with a large enough array of physical qubits.
This makes  surface code one of the most feasible QEC choices for demonstrating near-term FT quantum computation.

With a readily available surface code array, many recent research efforts have been put into improving the efficiency of FT quantum computation, varying from compilation~\cite{Ding2018MagicStateFU, Paler2019SurfBraidAC}, communication scheduling~\cite{JavadiAbhari2017OptimizedSC, Hua2021AutoBraidAF}, and micro-controller design~\cite{Tannu2017TamingTI}. All these works are based on a nontrivial assumption: we already find a scalable way to compose logical qubits on existing quantum devices, in particular SC devices, with the surface code family.

Yet, implementing surface code on SC devices itself is a difficult problem. %
The implementation of surface code separates physical qubits into two categories. 
The first category of qubits is called ``data qubits'' and is used to encode the logical qubit. 
Physical qubits in the second category are used to detect errors on data qubits, and are thus called ``syndrome qubits''. Each syndrome qubit extracts error syndromes on four neighboring data qubits with a specialized quantum  circuit, named ``measurement circuit''~\cite{fowler2012surface}. 
To make measurement circuits executable, surface code requires a 2D-lattice qubit array where each qubit is connected or coupled to four other qubits. Nonetheless, such an architecture is not readily available on many latest quantum processors~\cite{zhang2020high, gold2021experimental} as the dense connection in the 2D qubit array would induce a high physical error rate. 
Previous work tackles the gap between surface code and the sparse-connected SC device either by tailoring the architecture with tunable coupling~\cite{Google72Q} or by designing a QEC code upon the surface code~\cite{Chamberland2020TopologicalAS}. The former method is expensive and may introduce extra device noises while the latter method is not automated.

In this paper, we propose the first automatic synthesis framework ``\myCompilerName'' for stitching the surface code family to various SC devices, without any variation in device or surface code. Our framework builds upon recent theoretical work on generating individual measurement circuits over sparse-connected SC devices~\cite{Lao2020FaulttolerantQE, Chamberland2020TopologicalAS}. 
These works generalize the syndrome qubit to a tree of low-degree qubits (a.k.a bridge tree with constituent qubits called bridge qubits) that connects to the four target data qubits. In essential, these work trades the qubit degree with the qubit number.

Yet, these measurement circuit generation works are far from tackling the overall surface code synthesis. 
These works cannot generate measurement circuits unless the data qubits and bridge qubits for each measurement circuit are assigned. Even all measurement circuits are generated, how to efficiently execute them is still not answered by these works. Systematically, to tackle the surface code synthesis problem on SC devices, we should address three key challenges.
First, \textit{the allocation of data qubits}. 
If the four data qubits of a measurement circuit are far away from each other, many bridge qubits will be needed to connect these data qubits. Oppositely, if these data qubits are too close, there will not be enough space for a bridge tree. %
Second, \textit{the selection of bridge qubits}. An improper selection of bridge qubits may cause conflicts in measurement circuits that several measurement circuits contend for one bridge qubit.
Such bridge qubit conflict will cause a sequential execution of measurement circuits thus increasing the error detection latency and degrading the error correction performance.
Third,  \textit{the execution order of measurement circuits}. As indicated above, sequential execution of measurement circuits is not acceptable and
an efficient synthesis should utilize the parallelism between measurement circuits, as much as possible. %

Our framework decouples the optimization space of surface code synthesis with a modular optimization scheme. 
To be specific, our framework consists of three key steps. 

In the first step, we optimize the allocation of data qubits since %
data qubits are the key to gluing  measurement circuits together,
and once allocated their physical mapping should not change to avoid error proliferation. As long as allocated data qubit layout ensures the existence of measurement circuits, we prefer a shorter total distance between data qubits since this reduces bridge qubits overhead for constructing measurement circuits. To enable this optimization, we search data qubits over a series of ``rectangular'' blocks, with each block created by a pair of three-degree qubits or one four-degree qubit.
The key insight behind such design is that existing SC architectures can always be embedded into a 2D lattice and the four data qubits for a measurement circuit exactly form a rectangle in the 2D lattice.
Our design ensures the existence of measurement circuits since there are enough high-degree qubits in each rectangle to connect data qubits. Rectangle-based search is efficient as it enables a coarse-grained exploration over the search space. Furthermore, to reduce the possible conflicts of measurement circuits, we require the overlapping between rectangular blocks at most happens on rectangle boundaries.

In the second step, it naturally comes to the optimization of bridge qubits. The optimization goal is to reduce the conflicts between measurement circuits, i.e., reduce the mutual bridge qubits, and minimize the bridge tree connecting data qubits. The key insight for such optimization is that the size of the measurement circuit is proportional to the bridge tree size and measurement circuits without 
conflicts can be executed together to shorten the time window of error correction cycles. To meet these goals, we first limit the search scope of bridge qubits inside rectangular blocks enclosed by data qubits. Since rectangular blocks have zero overlapping areas, such local search could greatly reduce bridge qubit conflict. 
On the other hand, to find small bridge trees, we adopt two different 
heuristics which complement each other.
We then pick the best results from these two methods as bridge tree candidates.

In the third step, we optimize the order of measurement circuit execution. The reason for this optimization is that conflicts between measurement circuits can sometimes be inevitable. In such cases, we have to execute conflicted measurement circuits in sequential. To reduce the total time window of measurement circuits execution in this case, we propose an iterative refinement method. We first place conflicted measurement circuits into different partitions, and then refine the partitions by moving large measurement circuits into one partition while ensuring the compatibility between measurement circuits.
The key insight behind such optimization is that when executing measurement circuits in parallel, the execution time depends on the one with the largest circuit depth.  By moving large measurement circuits together, we can reduce the execution time of partitions consisting of small measurement circuits.

We evaluate the proposed synthesis framework by comparing it to manually designed QEC codes~\cite{Chamberland2020TopologicalAS} on two SC architectures. The results show the surface codes synthesized by our framework can achieve equivalent or even better error thresholds. This result is inspiring as it unveils the possibility that automated synthesis by machine can surpass the manual QEC code design of experienced theorists. 
We also investigate our framework on various mainstream SC architectures and present an analysis of several architecture design options based on the investigation result.
The proposed synthesis framework would be of great interest to both QEC researchers and quantum hardware designers.  
Theorists would have a baseline to compare with when designing novel QEC codes.
Hardware researchers could benefit from our framework in two aspects: a) they could focus on improving the device without worrying about meeting the requirements of a specific QEC code;
b) they can identify inefficient architecture designs for QEC codes with our synthesis framework.

Our contributions in this paper are summarized as follows:
\begin{itemize}
    \item We promote the importance of a synthesis framework for achieving good QEC implementations towards different quantum hardware architectures.
    \item We systematically formulate the surface code synthesis problem on SC devices for the first time and identify three key challenges: data qubit allocation, bridge qubit selection and  measurement circuit execution scheduling. %
    \item We design and implement a modular synthesis framework that tackles the identified three challenges step by step, with a series of insights extracted from surface code and SC architectures.
    \item Our evaluation demonstrates the effectiveness of the proposed synthesis framework with one comparative study to manually designed QEC codes and a comprehensive investigation on various mainstream SC architectures. 
\end{itemize}

\section{Background} 

In this section, we introduce the key concepts for understanding the implementation requirements of surface code~\cite{bravyi1998quantum,dennis2002topological,barends2014superconducting}. 
We do not cover the basics of quantum computing and recommend~\cite{nielsen2002quantum} for reference.

\begin{figure}[ht]
    \centering
    \includegraphics[width=\linewidth]{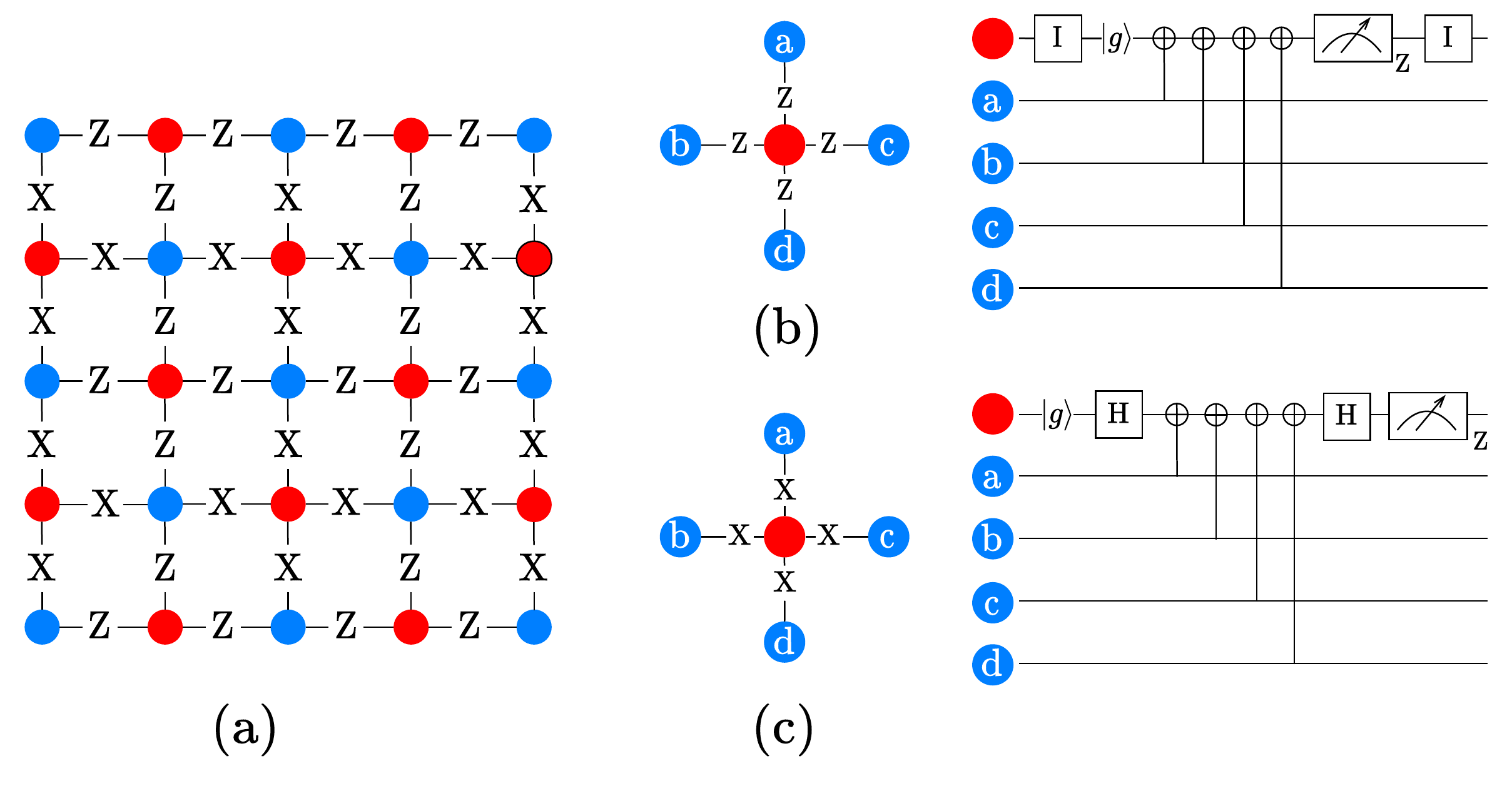}
    \caption{Common components of the surface code. (a) Surface code lattice with data qubits (blue dots) and syndrome qubits (red dots). (b) Z-type syndrome extraction and its circuit. (c) X-type syndrome extraction and its circuit.}
    \label{fig:surf_grid}
\end{figure}

\textbf{Data and syndrome qubits:}
Surface code encodes logical qubit information in a 2D lattice of physical qubits shown in Figure~\ref{fig:surf_grid}(a). Physical qubits in the code lattice can be divided into two types: data qubits and syndrome qubits (represented as blue and red dots respectively, see Figure~\ref{fig:surf_grid}(a)).
The logical information is encoded in the data qubits.
The error information on these data qubits can be extracted by applying surface code operations and measuring  syndrome qubits.

\textbf{Pauli operator and stabilizer:}
In surface codes, the relationship between one syndrome qubits and its neighboring data qubits is represented by the product of an array of Pauli operators (a.k.a Pauli string) labeled on the edges between data qubits and syndrome qubits (Figure~\ref{fig:surf_grid}(a)).
For each syndrome qubit, the Pauli string on its edges can be in one of two possible patterns. %
The first one (Z-type) is shown in Figure~\ref{fig:surf_grid}(b).
The connections between the center syndrome qubit and the four data qubits are all labeled by the $Z$ operators and represented by the Pauli string $Z_aZ_bZ_cZ_d$.
The second one (X-type) in Figure~\ref{fig:surf_grid}(c) is similar but all connections are labeled by the $X$ operators and represented by the Pauli string $X_aX_bX_cX_d$.
For these two different patterns, we will have corresponding syndrome measurement circuits to detect errors in the data qubits (shown on the right of Figure~\ref{fig:surf_grid}(b) and (c)).
The syndrome extraction circuits project the state of data qubits $\{a,b,c,d\}$ into the eigenstates of the corresponding Pauli strings.
In the context of QEC, the specific Pauli strings that one measures are called \textit{stabilizers}~\cite{gottesman1996class} and the syndrome extraction circuits are also known as \textit{stabilizer measurements} ~\cite{ calderbank1997quantum}.
Without ambiguity, we use the stabilizer notation to represent a syndrome extraction. And we denote the stabilizer $Z_aZ_bZ_cZ_d$ ($X_aX_bX_cX_d$) with $Z_{abcd}$ (respectively $X_{abcd}$) for simplicity.

\textbf{Error detection:} Surface code can detect Pauli X- and Z-errors on data qubits by using Z- and X-type stabilizer measurement circuits, respectively. 
An error on a data qubit may affect the measurement results of stabilizers associated with this data qubit.
By gathering all such stabilizer measurement results, a surface code error correction protocol can infer what errors occurred in the lattice and consequently apply the corresponding correction. 
Further details of the error correction protocol
can be found in \cite{fowler2012surface}.

\begin{figure}[ht]
    \hspace{-15pt}
    \centering
    \includegraphics[width=\linewidth]{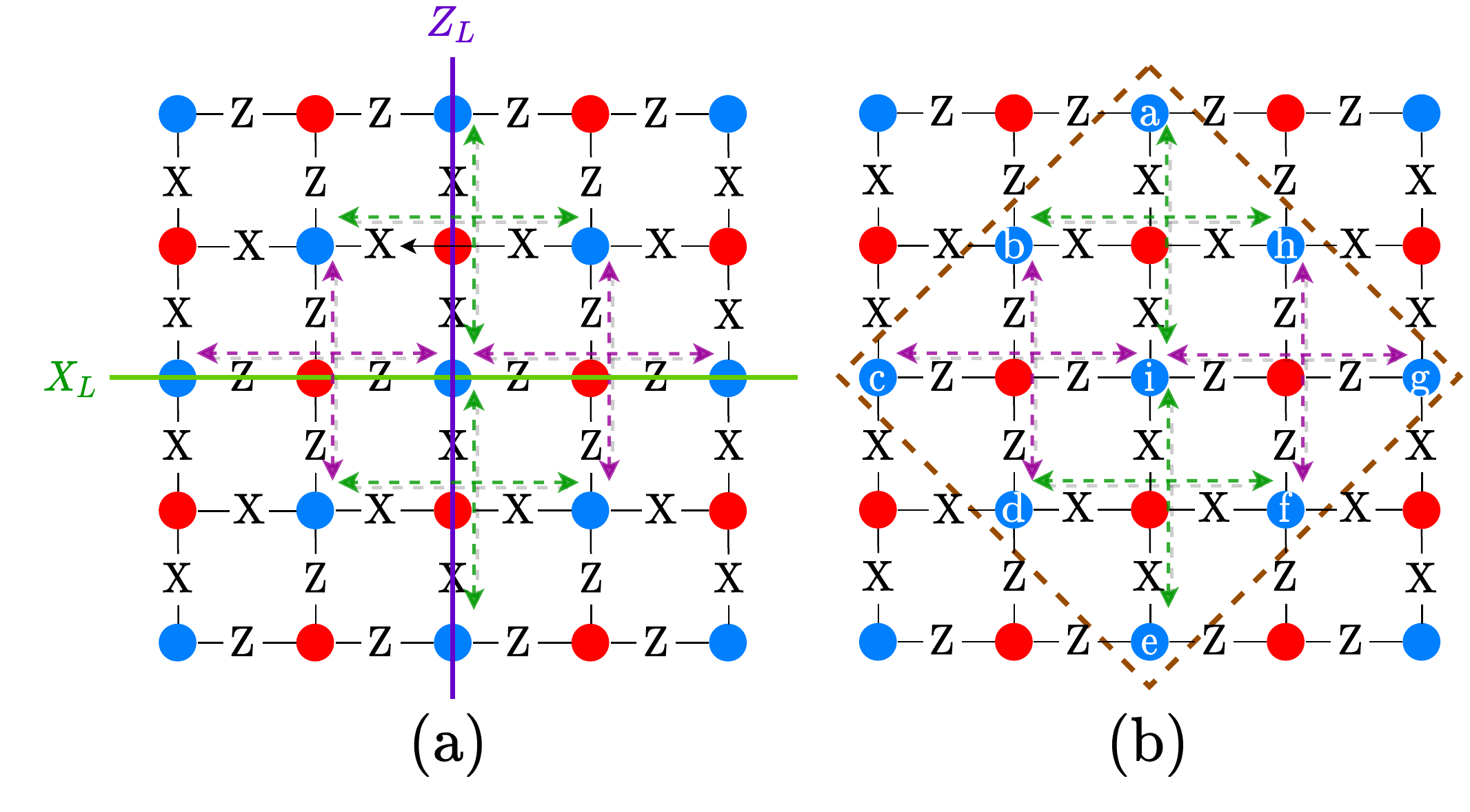}
    \caption{Code distance and compact lattice. (a)  Surface code logical operations with distance 3. (b) Inside the rotated rectangle (dashed brown line) is a compact surface code lattice with the same code distance.}
    \label{fig:surf_logi}
\end{figure}

\textbf{Code distance:} The error correction ability of QEC codes is measured by their code distance~\cite{knill2000theory,calderbank1996good}. The code distance is defined to be the minimum number of physical qubits in the support of logical X or Z operations on the encoded logical qubit
(denoted by $X_L$ or $Z_L$ in Figure~\ref{fig:surf_logi}). 
Figure~\ref{fig:surf_logi} (a) shows logical operations in a distance-3 surface code. It is possible to obtain a more compact surface code lattice without changing the code distance, as shown in Figure~\ref{fig:surf_logi}(b). In this paper, we focus on the rotated surface code in Figure~\ref{fig:surf_logi}(b).

\textbf{Stabilizer synthesis}: The first step of surface code synthesis is to implement individual stabilizer measurement circuit. Various methods have been proposed to map one single stabilizer to sparse-connected SC devices, such as the degree-deduction technique~\cite{Chamberland2020TopologicalAS}, and the flag-bridge circuit~\cite{Lao2020FaulttolerantQE, Chao2019FlagFE, Chamberland2017FLAGFE, Chamberland2019TriangularCC}. Once data qubits and needed ancillary qubits (including the syndrome qubit) for the stabilizer measurement are appointed, these method can generate the corresponding  measurement circuit.
Figure~\ref{fig:bridge_circuit_1} 
shows the generated flag-bridge circuit for stabilizer $Z_{bcid}$ with ancillary qubits $\{ e,s,f \}$.
These work only solved the low-level circuit generation problem of one stabilizer measurement, far from tackling the overall surface code synthesis, which is more than a simple collection of measurement circuits. %

\begin{figure}[ht!]
    \centering
    \includegraphics[width=0.75\linewidth]{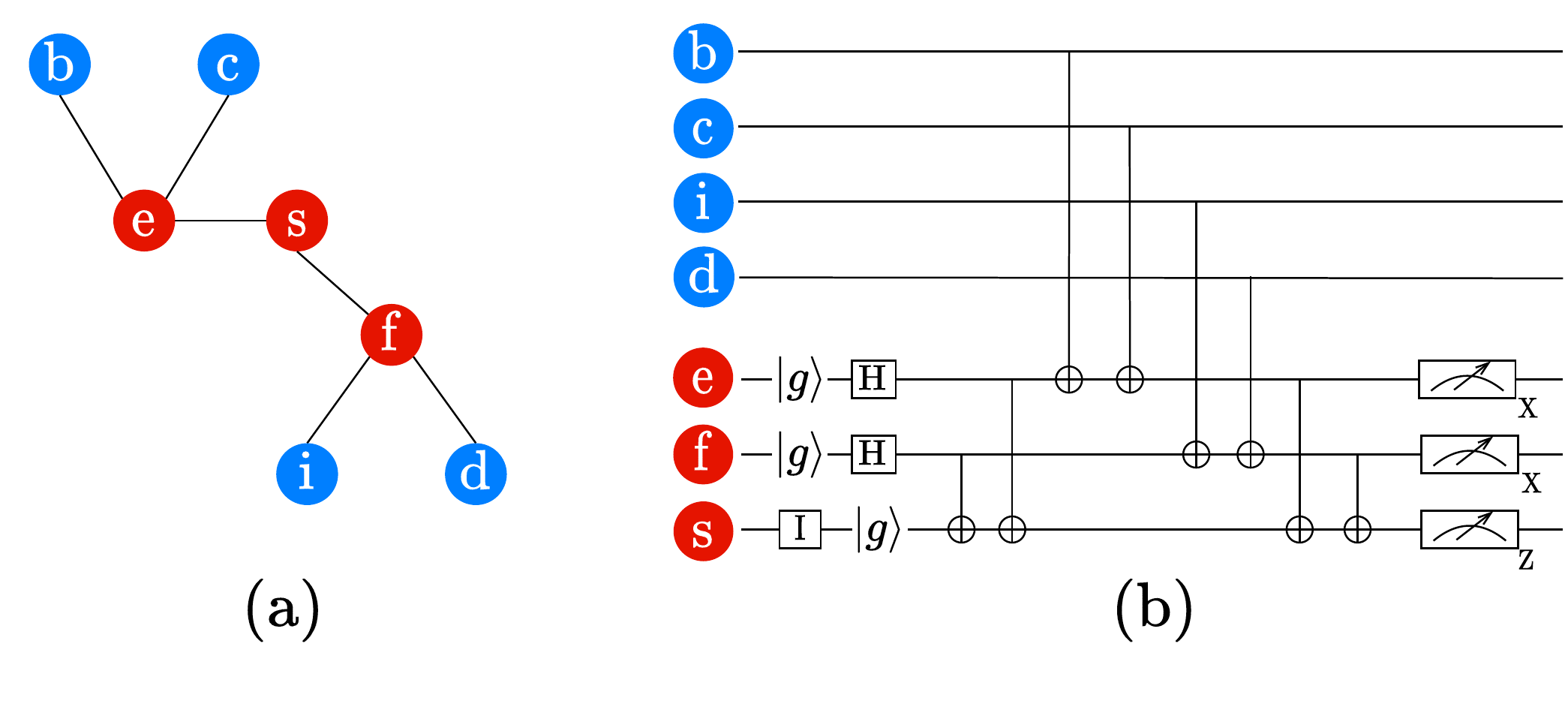}    %
    \caption{ Z-type stabilizer measurement circuit synthesis: (a) The coupling graph of data qubits (blue) and ancillary qubits (red). Qubit $s$ is the syndrome qubit. (b) The synthesized stabilizer measurement circuit that satisfies the coupling graph in (a). }\label{fig:bridge_circuit_1}
\end{figure}

Our framework targets at the high-level synthesis of the entire surface code and use these low-level stabilizer synthesis methods as the backend.

\begin{figure*}[ht!]
    \centering
    \includegraphics[width=0.85\linewidth]{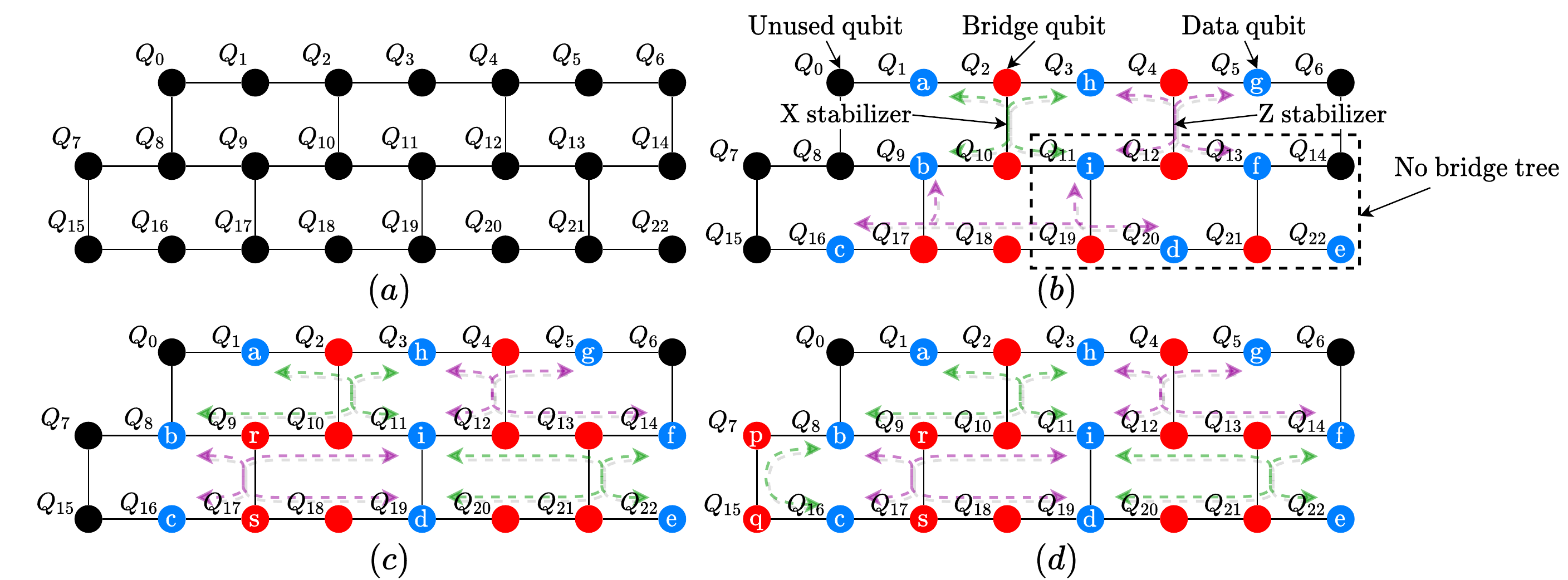}
    \caption{ A motivating example for synthesizing a (rotated) distance-3 surface code: (a) An SC device based on the hexagon structure; (b) A bad data qubit layout where the stabilizer $X_{idfe}$ cannot be measured; (c) A promising data qubit layout that ensures stabilizer measurement for all stabilizers; (d) An example for resolving bridge tree conflict. }
    \label{fig:motivation}
\end{figure*}

\section{Design Overview}
\label{sect: design_over}
The surface code relies on a 2D square-grid connectivity between physical qubits, while actual superconducting (SC) processors may not satisfy this requirement and then fail to execute the vanilla surface code.
In this paper, we aim to mitigate this gap by resynthesizing the surface code onto SC architectures with sparse connections.
We first formulate the surface code synthesis problem and then introduce the optimization opportunities.

\subsection{Surface code synthesis on SC architectures}

We consider implementing the (rotated) surface code in Figure~\ref{fig:surf_logi}(b) on a quantum device with a hexagon architecture (see Figure~\ref{fig:motivation}(a))~\cite{Chamberland2020TopologicalAS}. In this hexagon device, each qubit connects to at most three other qubits.
This imposes a challenge to synthesize stabilizer measurement circuits of a surface code since a syndrome qubit in either an X- or Z- type stabilizer measurement should connect to four data qubits (see Figure~\ref{fig:surf_grid}(b) and (c)).

To resolve this constraint, we can introduce ancillary qubits when synthesizing a stabilizer measurement circuit.
Suppose a mapping of the logical qubit in Figure~\ref{fig:surf_logi}(b) onto the hexagon architecture  is shown in Figure~\ref{fig:motivation}(b). 
To measure the stabilizer $Z_{bcid}$, one needs to connect the four data qubits $\{b,c,i,d\}$ with a syndrome qubit. 
In Figure~\ref{fig:motivation}(b), %
data qubits $\{b, c, i, d\}$ are mapped to physical qubits $\{Q_9, Q_{16}, Q_{11}, Q_{20}\}$, and %
the syndrome qubit is on $Q_{18}$.
To connect $Q_{18}$ and $\{Q_9, Q_{16}, Q_{11}, Q_{20}\}$, one can use the two ancillary qubits $Q_{17}$ and $Q_{19}$.
Qubits $Q_{17}$, $Q_{18}$ and $Q_{19}$ together form a tree that bridges the gap between data qubits $\{Q_9, Q_{16}, Q_{11}, Q_{20}\}$ and the syndrome qubit $Q_{18}$. Such a tree is called \textbf{bridge tree} and ancillae $\{Q_{17}, Q_{19}\}$ are called \textbf{bridge qubits}. For simplicity, we regard the syndrome qubit as a special bridge qubit.

\subsection{Optimization opportunities}

To deploy an entire surface code QEC protocol onto an SC architecture requires synthesizing a series of non-independent stabilizer measurements, which is far more complicated than handling one single stabilizer measurement.
In this section, we formulate the overall surface code synthesis problem into three key steps: data qubit allocation, bridge tree construction, and stabilizer measurement schedule.
We briefly introduce the objectives and the design considerations of each step.

\textbf{Data qubit allocation}: In this paper, we choose to allocate and fix the position of data qubits first as data qubits are the key to gluing stabilizer measurement circuits together and should not be changed once allocated.  
Comparing to data qubits, bridge qubits are \textit{dynamic} resources which are initialized, used, measured and decoupled in every QEC cycle, making them unsuitable for a pre-allocation.

The layout of the data qubits %
affects how efficiently stabilizer measurement circuits can be executed.
For example, we synthesize the (rotated) distance-3 surface code in Figure~\ref{fig:surf_logi}(b) with two data qubit layouts in Figure~\ref{fig:motivation}(b) and Figure~\ref{fig:motivation}(c).
In Figure~\ref{fig:motivation}(b), the stabilizer $X_{idfe}$ cannot be measured without moving data qubits and inserting SWAP gates, which are not allowed to avoid error proliferation.%
 In contrast, all stabilizer measurements ($\{X_{abhi}, X_{idfe}, X_{fg}, X_{bc}, Z_{bcid}, Z_{higf}, Z_{ah}, Z_{de} \}$) can be executed on Figure~\ref{fig:motivation}(c) since bridge trees for these stabilizers are readily available.

\textbf{Bridge tree construction}: 
After the data qubits are placed, the next step is to select the bridge qubits and construct bridge trees for stabilizer measurements.
The first constraint in this step is that we should minimize the number of bridge qubits since using more physical qubits 
results in larger measurement circuits which are naturally more error-prone.
The second constraint is that the construction of bridge trees affects the efficiency of error detection because two stabilizers can be simultaneously measured only if their bridge trees do not intersect.
For instance, referring to Figure~\ref{fig:motivation}(c), if we measure $X_{bc}$ with bridge qubits \{$r$, $s$\}, these two qubits cannot be used as bridge qubits in the measurement circuit of $X_{abhi}$ at the same time because the bridge qubits need to be reset at the beginning of any new measurement circuit.
However, if we measure $X_{bc}$ with bridge qubits \{$p$, $q$\} in Figure~\ref{fig:motivation}(d), we can measure $X_{bc}$ and $X_{abhi}$ in parallel. An efficient bridge qubit selection and tree construction should enable the concurrent measurement of as many stabilizers as possible.

\textbf{Stabilizer measurement scheduling}:
The third step is to schedule the execution of the stabilizer measurement circuits.
It would be desirable to execute the stabilizer measurement in parallel as much as possible since it can reduce the execution time and mitigate the decoherence error. 
However, stabilizer measurement circuits have overlapped bridge qubits (a.k.a bridge qubit conflict) cannot be executed simultaneously. %
For example in Figure~\ref{fig:motivation}(c), the measurement circuit of $X_{abhi}$ and $Z_{bcid}$ cannot be measured together since they share bridge qubits $\{q_{9}, q_{10}\}$. 
One possibility is to measure $X_{abhi}$ and $X_{idhe}$ first, then measures $Z_{hgif}$ and $Z_{bcid}$. Though seems promising, it is not optimal as these two groups of measurements take 20 operation steps in total, using the flag-bridge circuit~\cite{Lao2020FaulttolerantQE}(Figure~\ref{fig:bridge_circuit_1}) as backend. As a comparison, if we measure $X_{abhi}$ and $Z_{hgif}$ first and measure $X_{idfe}$ and $Z_{bcid}$ second, the total operation step number is only 18. 
It is usually not a simple task to schedule the stabilizer measurements optimally.
Our objective is to identify the potential parallelism in stabilizer measurements and figure out an efficient heuristic scheduling method to minimize the overall error detection cycle time (running all stabilizer measurements once is a cycle in the surface code QEC protocol).

\section{Synthesis Algorithm Design}
\label{sect:algorithm}
In this section, we introduce our surface code synthesis flow. As discussed above, we will introduce three key steps, data qubit allocation, bridge tree construction and stabilizer measurement scheduling.

\subsection{Data qubit allocator}
\label{sect:map_data}

\begin{algorithm}[h]\footnotesize
\SetAlgoLined
\KwIn{Device architecture graph $G$}
\KwOut{Data qubit layout $data\_layout$}

$L_h =$ all three- and four-degree nodes in $G$\;
$bridge\_rects = []$ ; \tcp{the set of bridge rectangles} 
\For{$n_a$ in $L_h$}{
    \uIf{$deg(n_a) == 3$}{
        $n_b$ = the nearest high-degree node of $n_a$\;
        $rect =$ the minimal rectangle containing $n_a$, $n_b$ and their neighboring qubits\;
    }
    \Else{
        $rect =$ the minimal rectangle containing $n_a$ and its neighboring qubits \;
    }
    $bridge\_rects.append(rect)$\;
}
$r_0 =$ the bridge rectangle at the top left corner of $G$\;

$bridge\_rect\_tuple = []$; \tcp{tuples of compatible bridge rectangles;}

\Repeat{$bridge\_rect\_tuple$ converges}{
    \For{$(r_1, r_2, r_3) \in \otimes^3{bridge\_rects}$}{ 
        \uIf{$r_0, r_1, r_2, r_3$ are mutually compatible}{
            $potent\_dqbits = $ qubits enclosed by $r_0, r_1, r_2, r_3$; \tcp{potential data area;}
            \uIf{$potent\_dqbits \ne \emptyset$}{
               $bridge\_rect\_tuple.append((r_0,r_1,r_2,r_3))$; \\
                break\;
            }
        }
    }
    set $r_0$ to $r_1, r_2, r_3$ in turn to find new combination of $r_0, r_1, r_2, r_3$ that has non-empty $potent\_dqbits$\;
}

$data\_layout = []$\;
\For{$r_0, r_1, r_2, r_3$ in $bridge\_rect\_tuple$}{
    $dqb$ = the qubit at the center of $potent\_dqbits$ of $r_0, r_1, r_2, r_3$\;
    $data\_layout.append(dqb)$\;
}
\caption{Data qubit allocation}
\label{alg:map_data_qubits}
\end{algorithm}

\begin{figure*}[ht!]
    \centering
    \includegraphics[width=0.9\linewidth]{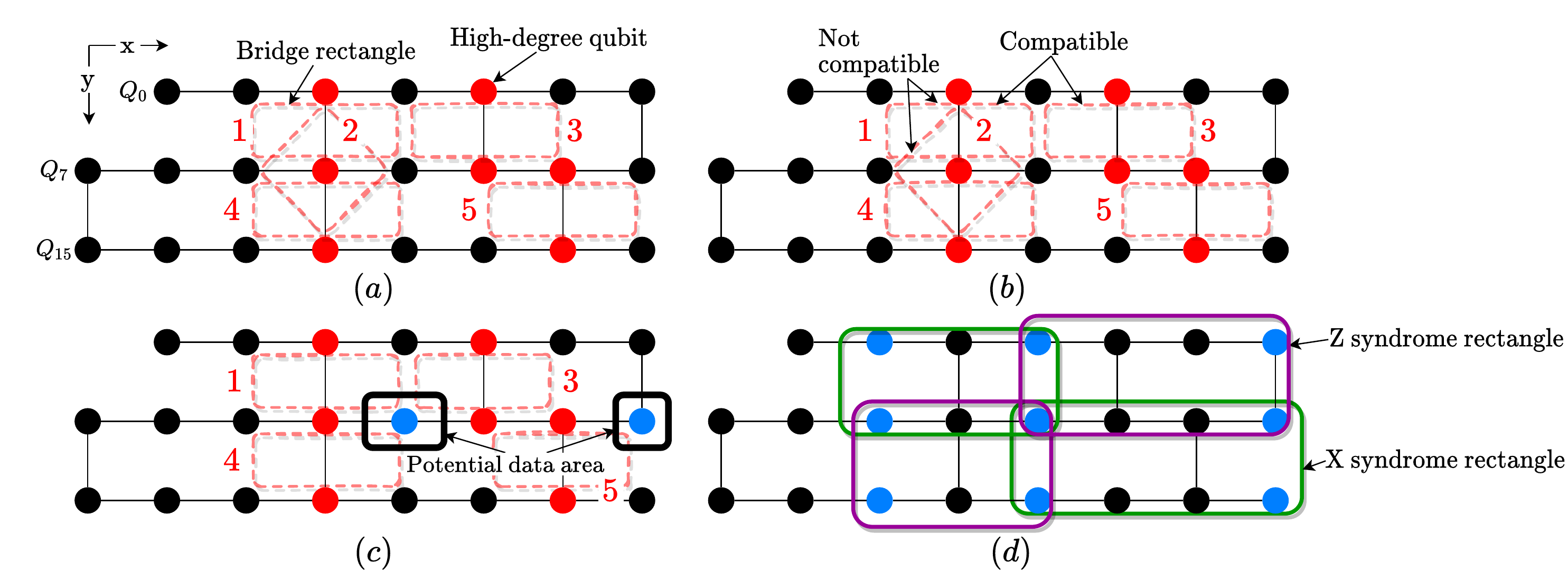}
    \caption{ Data qubit allocation example. (a) A modified device from Figure~\ref{fig:motivation}(a). Red circles indicate physical qubits with high degree of connectivity (i.e. with 3 or more edges). (b) Finding compatible bridge rectangles. (c) Locating data qubits. (d) Final data qubit layout and syndrome rectangles. }
    \label{fig:data_placement}
\end{figure*}

Since we decided to fix the location of data qubits, the basic requirement is that the measurement circuits can be feasibly constructed for all stabilizers. 
To ensure this property is guaranteed in the data qubit mapping we found, we have the following proposition regarding the bridge tree construction.
\begin{prop}\label{prop:bridge_tree}
In any bridge tree for a stabilizer having support on four data qubits, there are at least one four-degree node or two three-degree nodes.
\end{prop}
\begin{proof}
For any graph $G(V,E)$ where $V$ is the vertex set and $E$ is the edge set, we have $\sum_{v\in V}deg(v) = 2|E|$. An $n$-vertex tree always has $n-1$ edges. A bridge tree with four data qubits has four $1$-degree leaf nodes and all other nodes should have degree of at least $2$. Therefore we have $4 + \sum_{v\in V\backslash \text{data qubits}} deg(v) = 2n-2$ and $\sum_{v\in V\backslash \text{data qubits}} deg(v) = 2n-6 = 2(n-4) + 2$. We only have $n-4$ vertices after removing the four leaf nodes. So we must have at least one four-degree node or two three-degree nodes.
\end{proof}

Proposition~\ref{prop:bridge_tree} provides a necessary condition for a feasible data qubit layout. For each stabilizer to be executable, we should ensure that there are enough three-degree or four-degree qubits around data qubits of this stabilizer. We then introduce a graph-based data qubit layout where we can find bridge trees `locally'. A local bridge tree locates inside the spatial area bounded by the data qubits of a stabilizer. Such a bridge tree is promising as it often leads to a shallow  measurement circuit.
We use the SC device shown in Figure~\ref{fig:data_placement}(a) to illustrate the data qubit allocation algorithm. 
We embed the coupling graph of this SC device into a 2D grid so that all qubits can be referenced by the spatial coordinates on the plane.
Such embedding is possible as latest SC processors are usually designed in a modular structure.

Since high-degree nodes are critical, we %
keep a list (denoted by $L_h$) of all three- and four-degree nodes in the device grid and record their coordinates. 
In Figure~\ref{fig:data_placement}(a), $L_h = \{ Q_{2}, Q_{4}, Q_{10}, Q_{12}, Q_{13}, Q_{18}, Q_{21} \}$. 
We then process the list of high degrees sequentially.
For the next %
node $n_a$ in $L_h$, if it is a three-degree node, we search for its nearest high-degree node $n_b$ and then create a minimal rectangle that contains $n_a$ and $n_b$ as well as their neighboring qubits.  If $n_a$ is node with degree $\ge 4$, then we create a rectangle that contains $n_a$ and its neighboring qubits. The rectangle created at this step is called ``bridge rectangle''. 
Figure~\ref{fig:data_placement}(a) depicts five bridge rectangles resulted from $\{Q_2, Q_{10}\}$, $\{ Q_{10} \}$, $\{Q_4, Q_{12}\}$, $\{Q_{13}, Q_{21}\}$ and $\{Q_{18}, Q_{10}\}$ , and we index them from 1 to 5. There are also other bridge rectangles that can be identified, but we omit them here for simplicity.

We can now determine the position of data qubits by using bridge rectangles. 
As shown in Figure~\ref{fig:surf_grid}, each data qubit is shared by four stabilizers. Thus, we can fix the position of a data qubit with four bridge rectangles. We begin with rectangle 1 to find four compatible bridge rectangles.
We can also begin with rectangle 2 which is created from a four-degree qubit. This will create a different surface code synthesis and we will discuss it in Section~\ref{sect:evaluation}. Here we focus on rectangle 1. Two bridge rectangles are said to be compatible if their intersection area is zero. As we can see in Figure~\ref{fig:data_placement}(b), rectangle 2 is not compatible with rectangle 1 and rectangle 4, while rectangles 1, 3, 4, and 5 are mutually compatible. We do not use incompatible rectangles because they may not allow a feasible data qubit mapping.
We then search for data qubits in the potential data area (black  rectangle in Figure~\ref{fig:data_placement}(c)) which is enclosed by four compatible bridge rectangles, as shown in Figure~\ref{fig:data_placement}(c). If the potential data area is empty, we choose another four compatible bridge rectangles. Otherwise,
we pick the qubit at the center of the potential data area as a data qubit. 

For boundary cases, we may not have enough bridge rectangles to locate one data qubit. For example, for rectangle 3, its bottom right corner is only neighbored by rectangle 5. In this case, we should locate the data qubit with constraints only from these two bridge rectangles. Specifically, a potential data qubit should satisfy: A) its x axis value $\ge$ the largest x axis value in rectangle 3 and 5; B) its y axis value should lie between the largest y axis value of rectangle 3 and the smallest y axis value of rectangle 5.  With these spatial constraints, the only qubit we can find is $Q_{14}$, as shown in Figure~\ref{fig:data_placement}(c). Other data points are found in a similar way. 

The final layout for data qubits and syndrome rectangles defined by data qubits are shown in Figure~\ref{fig:data_placement}(d). 
A syndrome rectangle is 
the extension of the bridge rectangle to include  allocated data qubits. We can then assign a stabilizer to each syndrome rectangle and synthesize the corresponding measurement circuits locally (using qubits inside each syndrome rectangle). The syndrome rectangle ensures the existence of local bridge trees and enables a natural data qubit assignment for the assigned stabilizer. For the stabilizer $X_{abcd}$ (or $Z_{abcd}$) in Figure~\ref{fig:surf_grid},
we can simply assign the topmost data qubit in a syndrome rectangle to be the data qubit `a', with the leftmost, rightmost, and bottom data qubit being data qubits `b', `c', `d', respectively. In the next section, we will discuss how to find an efficient bridge tree to connect data qubits $\{a,b,c,d\}$.

\subsection{Bridge tree finder}
\begin{algorithm}[h]\footnotesize
\SetAlgoLined
 
\KwIn{A syndrome rectangle $R$ with data qubits $\{a,b,c,d\}$}
\KwOut{Candidate bridge trees}

$star\_trees = []$; \tcp{bridge trees generated by the star tree method;}
$branching\_trees = []$; \tcp{bridge trees generated by the star tree method;}

\For{$qb$ in $R$} {
    $T$ = the bridge tree by connecting qubit $qb$ to data qubits $\{a,b,c,d\}$ with shortest paths\;
    update $star\_trees$ s.t. it only contains trees no larger than $T$\;
}

let $\{a',b',c',d'\}$ be a arrangement of $\{a,b,c,d\}$ s.t.  $l_{a'b'} + l_{c'd'} = \min \{ l_{ab} + l_{cd}, l_{ac} + l_{bd}, l_{ad} + l_{bc} \}$\; \tcp{$l_{ab}$ is the distance of $a\to b$;}

connect $a'$ and $b'$, $c'$ and $d'$ with shortest paths, respectively\;

connect the path $a' \to b'$ and $c' \to d'$ with shortest paths\;

\For{$qb_1$ in $a' \to b'$, $qb_2$ in $c' \to d'$} {
    $T$ = the bridge tree by connecting $qb_1$ and $qb_2$ by shortest paths\;
    update $branching\_trees$ s.t. it only contains trees no larger than $T$\;
    
}
Merge $star\_trees$ and $branching\_trees$ to find a list of minimal bridge trees;

\caption{Bridge tree finder}
\label{alg:bridge_tree_finder}
\end{algorithm}

\begin{figure}[ht!]
    \centering
    \includegraphics[width=\linewidth]{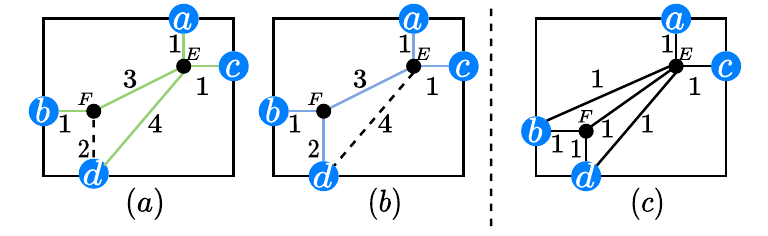}
    \caption{ Finding bridge trees in a syndrome rectangle with data qubits $\{a,b,c,d\}$. (a)(b) shows the case where path merge is efficient, while (c) shows when path merge incurs extra overhead. (a) Green edges denote the shortest paths from qubit $E$ to data qubits and these paths form a bridge tree with length 10. (b) Blue edges form a bridge tree with length 8. (c) An example where data qubits are close to each other.}
    \label{fig:bridge_path}
\end{figure}

After locating the data qubits and syndrome rectangles, we then construct the bridge tree for each syndrome rectangle.
Since syndrome rectangles only intersect at borderlines, strict local bridge trees, whose bridge qubits all locate within the interior of the syndrome rectangle, will enable concurrent measurement of stabilizers naturally. Thus, we only consider finding bridge trees in the syndrome rectangle. %
To enable higher stabilizer measurement fidelity, we prefer using small bridge trees. 
Actually, the error correction of the synthesized code is sensitive to the length of the bridge tree because one more edge in the bridge tree will result in two more CNOT gates in the measurement circuit and likely leads to a high probability of correlated errors which are hard to detect and correct. Thus, 
bridge trees should be as small as possible.

One natural way for finding small bridge trees is to first locate the bridge tree root and then connect the tree root to data qubits with shortest paths. We denote this method by ``{star tree}'' method. One drawback of this method is that it may miss the opportunity of path merge.
 For example, in the syndrome rectangle in Figure~\ref{fig:bridge_path}(a), the bridge tree induced by the shortest paths from interior qubit $E$ to data qubits has a length 10 (green edges) . In contrast, by merging paths $E\to F\to b$ and $E \to d$, we can get a bridge tree of length 8 (blue edges in Figure~\ref{fig:bridge_path}(b)) , which reduces the number of CNOT gates by at least 4 in the resulted stabilizer measurement circuit.

To overcome the problem above, we propose the ``{branching tree}'' method, which first connects close data qubit pairs by shortest paths and then connects these shortest paths to build a complete bridge tree (pseudo code in in Algorithm~\ref{alg:bridge_tree_finder}). Suppose we are constructing the bridge tree for the syndrome rectangle in Figure~\ref{fig:bridge_path}(a). We first find the shortest paths $a \to c$ and $b \to d$, since $l_{ac} + l_{bd}$ ($l_{ac}$ is the length of the shortest path from $a$ to $c$) is shorter than $l_{ab} + l_{cd}$ and $l_{ad} + l_{bc}$. We then connect path $a\to c$ and $b \to d$ with the path $E \to F$,  immediately resulting in the small bridge tree (blue edges) in Figure~\ref{fig:bridge_path}(b). The following proposition bounds the length of the bridge tree generated by the branching tree method:

\begin{prop}\label{prop:bridge_perf}
Let the total edge length of the bridge tree $T$ generated by the branching tree method be $E(T)$, then,
$$E(T) \le \frac{1}{2}(l_{ab}+l_{ac}+l_{ad}+l_{bc}+l_{bd}+l_{cd}).$$
\end{prop}
\begin{proof}
W.l.o.g., we assume $l_{ab}+l_{cd} \le \min \{l_{ac} + l_{bd}, l_{ad}+l_{bc}\}$. Then in $T$, we first connect $a$ and $b$, $c$ and $d$, respectively.
On the other hand, the distance between shortest paths $a\to b$ and $c \to d$ is smaller than $\min \{l_{ac}, l_{ad}, l_{bc}, l_{bd}\}$. This proposition then can be proved by combining these two inequality. %
\end{proof}

Generally, the branching tree method is more efficient if $\min \{ l_{ab}+l_{cd}, l_{ac}+l_{bd}, l_{ad}+l_{bc}  \}$ is small, as shown in Figure~\ref{fig:bridge_path}(a)(b). In this case, the length of the resulted branching tree is very close to $\frac{1}{2} ( l_{ad} + l_{bc} )$. Instead, the length of a star tree is at least $\max \{ l_{ad}, l_{bc} \} + 2$, thus leading to a larger bridge tree. On the other hand, if $\max \{ l_{ab}+l_{cd}, l_{ac}+l_{bd}, l_{ad}+l_{bc} \}$ is small, the benefit of path merge may not cancel out its overhead for not using shortest paths. Figure~\ref{fig:bridge_path}(c) shows an example where the star tree has shorter length.
In practice, we will run both star tree method and branching tree method and find small bridge trees by merging the result of these two methods, as shown in Algorithm~\ref{alg:bridge_tree_finder}. Once the bridge tree is determined, we can assign the syndrome qubit to the center node of the bridge tree.

Overall, Algorithm~\ref{alg:bridge_tree_finder} can generate small bridge trees that approximate the optimal bridge tree as long as the distance between data qubits is small %
(by Proposition~\ref{prop:bridge_perf}). 
Another feature of Algorithm~\ref{alg:bridge_tree_finder} is that it may find many different bridge trees for a stabilizer since the shortest paths between nodes are not unique. This feature provides some flexibility for the stabilizer scheduling discussed in the next section.

\subsection{Stabilizer measurement scheduler}

\begin{algorithm}[h]
\SetAlgoLined
\footnotesize
\KwIn{Binary tuples of stabilizer and syndrome rectangle: $\{ (s, R) \}$.}
\KwOut{A schedule $P$ of binary tuples of stabilizer and bridge trees.}
\tcp{Schedule initialization;}
$S_1 = $ tuples of X-stabilizers and syndrome rectangles\;
$S_2 = $ tuples of Z-stabilizers and syndrome rectangles\;
\uIf{$exec\_time(S_1) < exec\_time(S_2)$}{
    swap($S_1$, $S_2$)\;
}
\tcp{Iterative refinement;}
\Repeat{$S_1$ converges}{
    $r_2$ = $(s, R)$ in $S_2$ that has longest execution time\;
    $swap\_list = [r_2]$\;
    \For{$i$ in $[0:k]$}{
        $S = S_{i\%2+1}$\;
        \For{$r$ in $swap\_list$}{
            $swap\_list.remove(r)$; $S.append(r_2)$\;
            \For{$r_1$ in $S$ in descending order }{
                \uIf{$r_1$ and $r$ does not have compatible bridge trees}{         \uIf{$exec\_time(r_1)>exec\_time(r)$}{
                            terminate the refinement loop\;
                    }
                    $swap\_list.append(r_1)$\;
                    $S.remove(r_1)$;
                }
            }
            \uIf{$swap\_list == \emptyset$}{
                break\;
            }
        }        
    }
    \uIf{$swap\_list \ne \emptyset$ }{
        recover $S_1$ and $S_2$ to the values before this iteration\;
        break;
    }
}
\caption{Iterative stabilizer scheduler}

\label{alg:stabilizer_schedule}
\end{algorithm}

The order of stabilizer measurements affects the time required by the error detection protocol. 
Our goal is to reduce the time requirement for all stabilizer measurements since it will naturally reduce the decoherence error during the process of stabilizer measurements.
To achieve this goal, we need to measure as many stabilizers as possible in parallel. Yet the fact that only compatible bridge trees that do not have common bridge qubits can be measured simultaneously imposes a constraint for  stabilizer measurement scheduling: only stabilizers that have compatible bridge trees can be measured in the same time. A valid schedule for the order of stabilizer measurements should avoid executing two ``conflicted'' stabilizers together.
To satisfy the constraint of stabilizer measurement scheduling yet exploit the parallelism in stabilizer measurements, we propose a heuristic scheduling approach in Algorithm~\ref{alg:stabilizer_schedule}, which consists of two steps: schedule initialization and refinement loop.

\paragraph{Schedule initialization} the proposed data qubit allocation ensures that syndrome rectangles of the same type do not have bridge tree conflicts, i.e., the measurements of X (or Z) stabilizers are compatible with each other. We then initialize the stabilizer measurement schedule by two sets $S_1$ and $S_2$ which contain X- and Z-type stabilizers, respectively.

\paragraph{Refinement Loop} the core idea of the refinement loop in Algorithm~\ref{alg:stabilizer_schedule} is to move stabilizers with large bridge trees into one set. The motivation for such refinement is that the execution time for a set of stabilizers is determined by the stabilizer with the deepest measurement circuit. With the refinement loop, except one stabilizer set which contains stabilizers with 
large bridge trees, remaining stabilizer sets only include stabilizers with small bridge trees and can be measured in a short time. %

To illustrate how the Algorithm~\ref{alg:stabilizer_schedule} works, suppose we are given stabilizers and syndrome rectangles in Figure~\ref{fig:stabilizer_schedule}. Initially, we have $S_1 = \{(s_1, R_1), (s_4, R_4), (s_5, R_5)\}$ and $S_2 = \{(s_2, R_2), (s_3, R_3), (s_6, R_6)\}$. We then send the largest element in $S_2$, which is $(s_2, R_2)$ in this case, to the $swap\_list$ and swap it into $S_1$. Since $(s_4, R_4)$ and $(s_2, R_2)$ do not have compatible bridge trees, we will move $(s_4, R_4)$ to $S_2$.
In $S_2$, $(s_6, R_6)$ is not compatible with $(s_4, R_4)$, so it will be swapped into $S_1$. 
After this swap, the refinement loop will stop since the $swap\_list$ is empty and every stabilizer in $S_1$ has a larger bridge tree than the stabilizer in $S_2$. The resulted stabilizer schedule is shown in Figure~\ref{fig:stabilizer_schedule}(b). Comparing to the initial schedule, the refined schedule in Figure~\ref{fig:stabilizer_schedule}(b) reduces the error detection cycle by one time step, and reduces the CNOT gate number by two.

The stabilizer measurement schedule found by the proposed heuristic should be better than the schedule where different types of stabilizers are measured in different rounds since every successful refinement iteration returns a better stabilizer measurement schedule.

\begin{figure}[ht!]
    \centering
    \includegraphics[width=\linewidth]{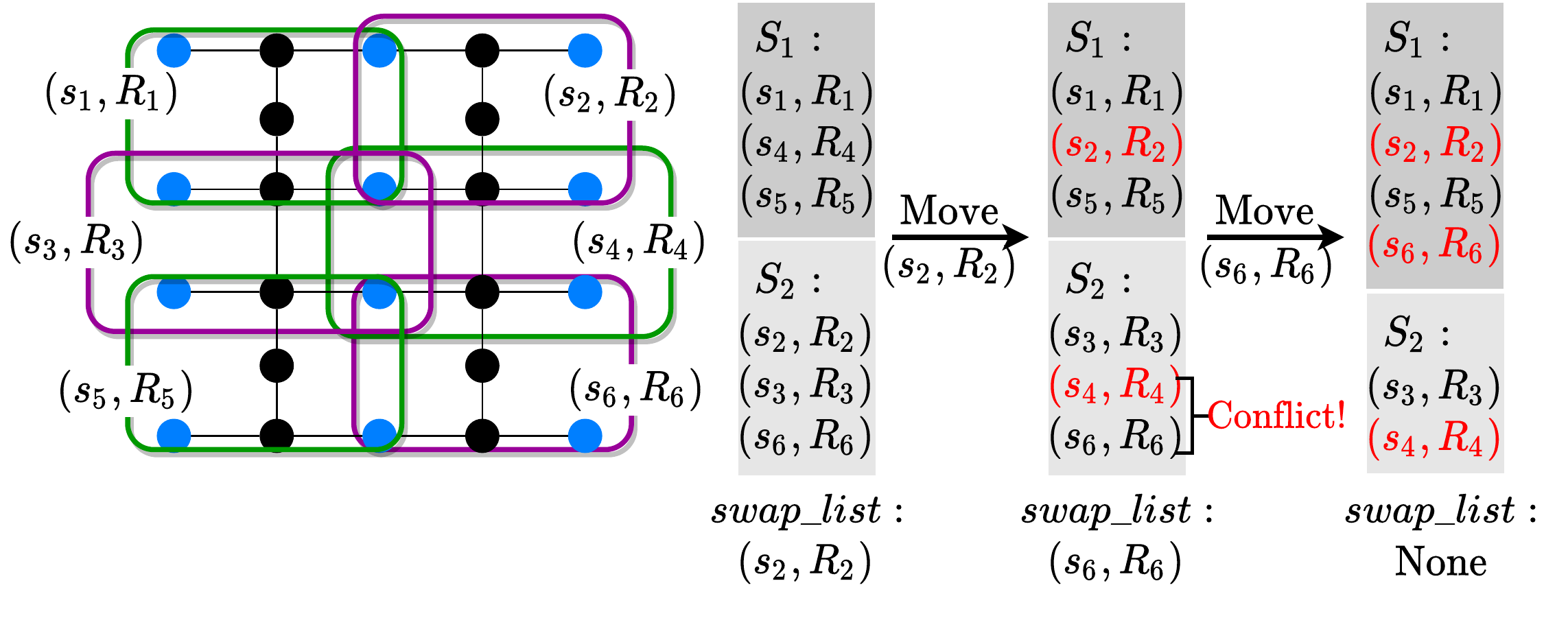}
    \caption{ An example of stabilizer measurement scheduling. } %
    \label{fig:stabilizer_schedule}
\end{figure}
\section{Evaluation}
\label{sect:evaluation}

In this section, we first evaluate the proposed synthesis framework ``\myCompilerName'' by comparing the generated surface codes with state-of-the-art manually designed QEC codes. 
We then demonstrate the efficiency of the proposed synthesis framework  by analyzing the error correction performance and resource overhead of the synthesized surface codes on mainstream SC architectures.

\subsection{Experiment Setup}

\paragraph{Evaluation setting}
We use the flag-bridge circuit~\cite{Lao2020FaulttolerantQE} as the backend for instantiating stabilizer measurement circuits. We select this backend because it provides the extra feature of fault-tolerant error correction.
We implement all numerical simulations with stim v1.5.0, which is a fast stabilizer circuit simulator~\cite{gidney2021stim}. We use PyMatching v0.4.0~\cite{higgott2021pymatching} for error decoding with measurement signals from bridge qubits. Error rates are computed by performing $10^5$ simulations, on a Ubuntu 18.04 server with a 6-core Intel E5-2603v4 CPU and 32GB RAM.

\paragraph{Metrics} We evaluate the \textbf{error threshold} of the synthesized surface codes to demonstrate their error correction performance. Error threshold indicates what hardware error rates can be tolerated and a higher error threshold is preferred. 
\textbf{Time-step counts} in an error detection cycle can also 
indicate the error correction performance~\cite{fowler2012surface}. The time-step counts also determines the execution speed of the surface code. A small time-step count is preferred. 
Finally, 
We evaluate the resource requirement of the synthesized surface codes with \textbf{CNOT counts} and \textbf{qubit counts}. A resource-efficient synthesis should use less CNOT gates and bridge qubits.

\paragraph{Device Architectures} %
we use two categories of device architectures. The architectures are shown in Table~\ref{tab:device_arch}. The first category architectures are built by tiling polygons and serve as basic structures for many SC devices, \emph{e.g.} Google's Sycamore~\cite{Arute2019QuantumSU} and IBM's latest machines~\cite{Jurcevic2020DemonstrationOQ}. The second category architectures consist of `heavy' architectures which insert one qubit for each edge of polygon devices. Edges with one extra qubit in the middle are called ``heavy edges''.
Heavy architectures have lower average qubit connectivity because of the inserted two-degree qubits. Heavy architectures are used by IBM devices~\cite{Jurcevic2020DemonstrationOQ}. Square and heavy square architectures can be embedded into a 2D grid naturally. Hexagon and heavy hexagon architectures can be embedded into a 2D grid by squashing the hexagon into the shape of a rectangle.

\begin{table}[h]
\centering
\caption{Overview of device architectures.}
\label{tab:device_arch}
\resizebox{0.49\textwidth}{!}{
\begin{tabular}{|m{1cm}|m{1.2cm}|m{1cm}|c|m{2.5cm}|}
\hline
Type                          & Name          & Building blocks & Tilling Example & Remark \\ \hline

\multirow{2}{1cm}{\raisebox{-36pt}{\parbox[c]{1cm}{Polygon Arch}}} & \raisebox{-16pt}{\parbox[c]{1cm}{Square}}       &  \raisebox{-22pt}{\hspace{6pt}\includegraphics[height=3\fontcharht\font`\B]{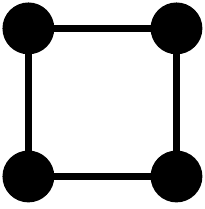}}               &  \raisebox{-34pt}{\includegraphics[height=6\fontcharht\font`\B]{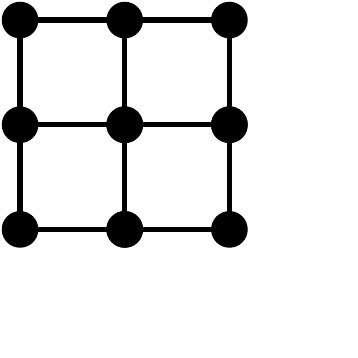}} &  \raisebox{-14pt}{\parbox[c]{2.5cm}{Each square can have at most four neighboring squares for tiling.}}      \\ \cline{2-5} 
& \raisebox{-0pt}{\parbox[c]{1cm}{Hexagon}}     &  \raisebox{-0pt}{\hspace{5pt}\includegraphics[height=4\fontcharht\font`\B]{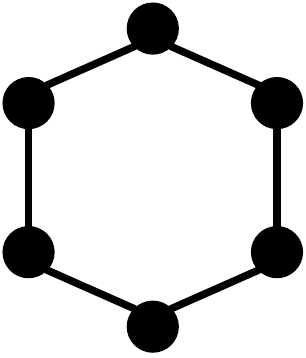}
}               & \raisebox{-14pt}[24pt]{\includegraphics[height=6\fontcharht\font`\B]{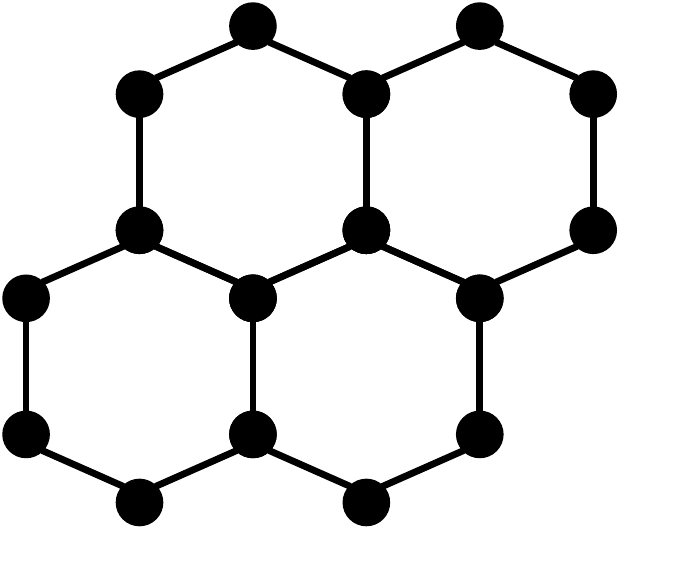}}  & 
\raisebox{16pt}{\parbox[c]{2.5cm}{Each hexagon can have at most six neighboring hexagons for tiling.}}      \\ \hline

\multirow{2}{1cm}{\raisebox{-24pt}{\parbox[c]{1cm}{Heavy Arch}}}   & \raisebox{12pt}{\parbox[c]{1cm}{Heavy square}}  & \raisebox{8pt}{\hspace{5pt}\includegraphics[height=3.5\fontcharht\font`\B]{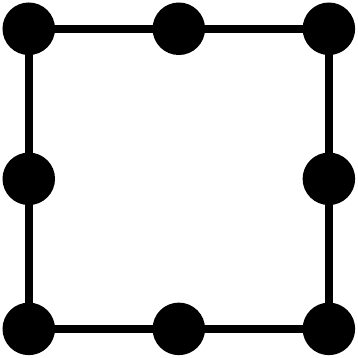}}                &  \raisebox{-10pt}[26pt]{\includegraphics[height=5.8\fontcharht\font`\B]{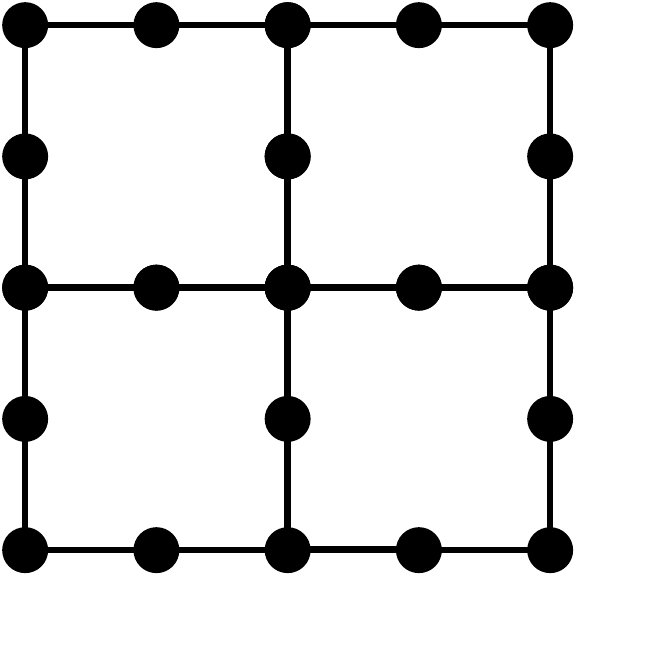}}  & \raisebox{12pt}{\parbox[c]{2.5cm}{Heavy squares are tiled like squares.}}    \\ \cline{2-5} 
& \raisebox{14pt}{\parbox[c]{1cm}{Heavy hexagon}} & \raisebox{5pt}{\hspace{2pt}\includegraphics[height=8\fontcharht\font`\B]{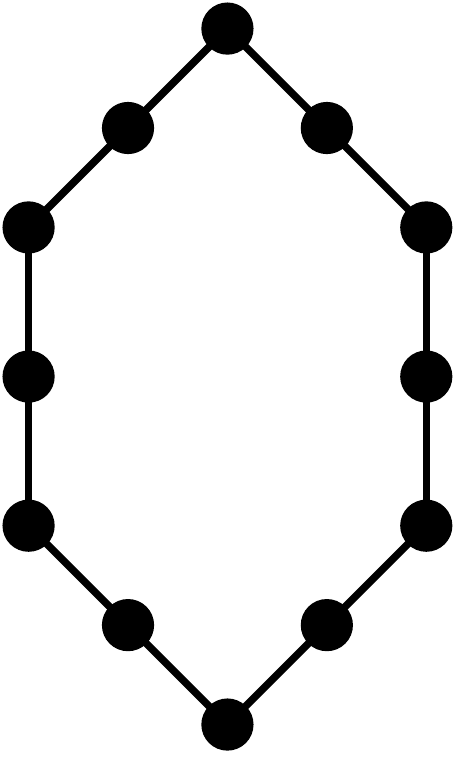}}                & \raisebox{-19pt}[36pt]{\includegraphics[height=10\fontcharht\font`\B]{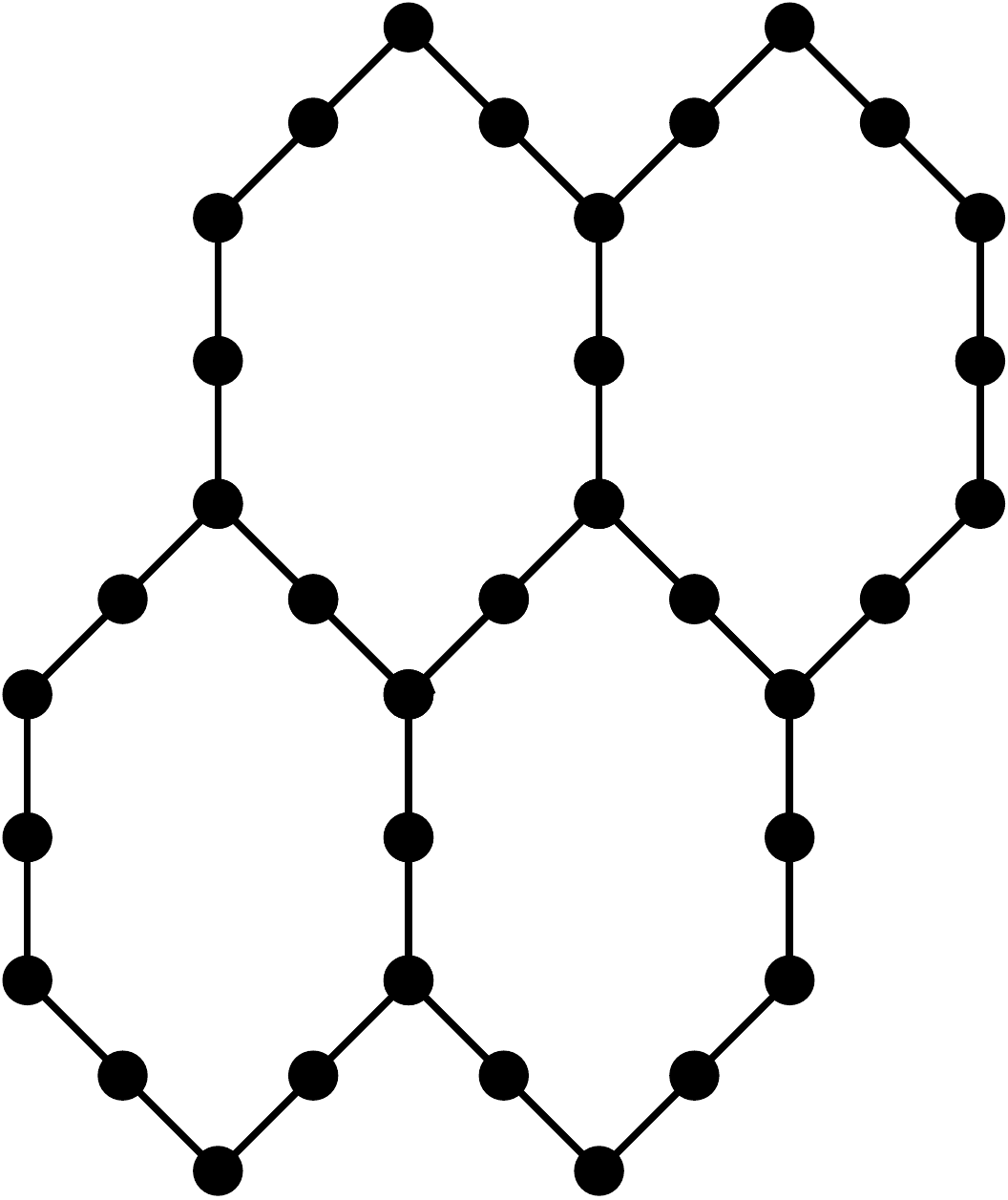}}   & \raisebox{14pt}{\parbox[c]{2.5cm}{Heavy hexagons are tiled like hexagons.}}    \\ \hline
\end{tabular}%
}

\end{table}

\paragraph{Error model} In all simulations, we assume the following circuit-level error model:
For the gate error, we assume an error probability $p_e$ for the single-qubit depolarizing error channel on single-qubit gates, the two-qubit depolarizing error channel on two-qubit gates and the Pauli-X error channel on measurement and reset operations. For the idle decoherence error, we assume each idle qubit is followed by a single-qubit depolarizing error channel per gate duration 
with error probability $0.0002$, which is estimated by the decoherence error formula $1-e^{-\frac{t}{T}} \approx 0.0002$, with the gate execution time $t= 20\,n s$  and the relaxation or dephasing time $T= 100\,\mu s$. %

These errors happen on all qubits, including data qubits and bridge qubits.

\subsection{Comparing to manually designed QEC code}

We first compare our synthesized surface code to the two manually designed QEC codes by Chamberland et al.~\cite{Chamberland2020TopologicalAS} on 
heavy architectures.
Figure~\ref{fig:our_ibm}(a)(b) shows the qubit layout and stabilizer measurement circuits of our synthesized surface codes on the heavy square architecture (`\myCompilerNameSpace Heavy Square')  and the heavy hexagon architecture (`\myCompilerNameSpace Heavy Hexagon'),
Figure~\ref{fig:our_ibm}(c)(d) shows the manually designed QEC codes on the heavy square architecture (`IBM Heavy Square') and  the heavy hexagon architecture (`IBM Heavy Hexagon'). 
The error thresholds of these codes are in  Figure~\ref{fig:comp_ibm}.

Overall, compared with the manually and specifically designed codes on the two architectures, the surface codes synthesized by \myCompilerNameSpace can have comparable or even better performance.
On the heavy hexagon architecture, the error threshold of `\myCompilerNameSpace Heavy Hexagon' is 0.0033\% which is 106\% higher than that of `IBM Heavy Hexagon' (0.0016\%), as shown in Figure~\ref{fig:comp_ibm}(a).
Such benefit comes from the fact that the `IBM Heavy Hexagon' code uses the Bacon-Shor scheme for Pauli Z-error correction (Figure~\ref{fig:our_ibm}(d)) which is not as effective as the surface code. %
On the heavy square architecture, the error threshold of `\myCompilerNameSpace Heavy Square' is the same as that of `IBM Heavy Square', as shown in Figure~\ref{fig:comp_ibm}(b).  This is because the code synthesized by \myCompilerNameSpace is 
 almost identical to that of `IBM Heavy Square' except stabilizers on boundaries, as shown in Figure~\ref{fig:our_ibm}(a)(c). %
In summary, \myCompilerNameSpace can automatically generate QEC codes that have similar or even better error correction performance compared with manually designed codes on the two studied architectures.

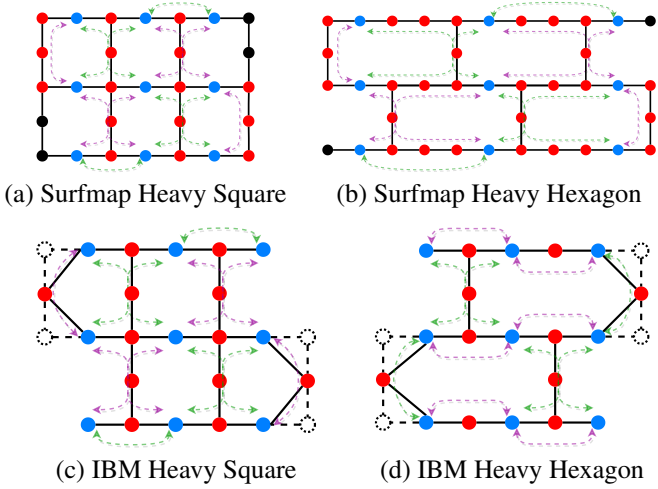
\begin{figure}
    \centering
    \begin{tikzpicture}
\node [above right,inner sep=0] (image) at (0,0) {\includegraphics[height=0.12\textwidth]{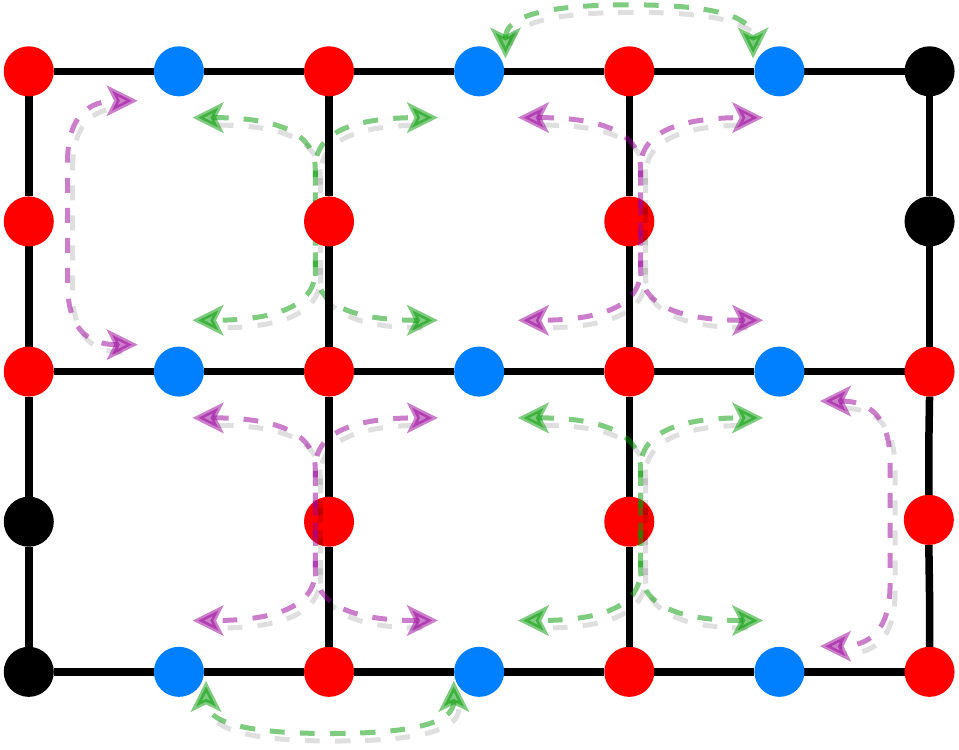}};
\node [below=0.01cm of image]  {(a) \myCompilerNameSpace Heavy Square};

\node [above right,inner sep=0] (image1) at (3.8,0) {\includegraphics[height=0.12\textwidth]{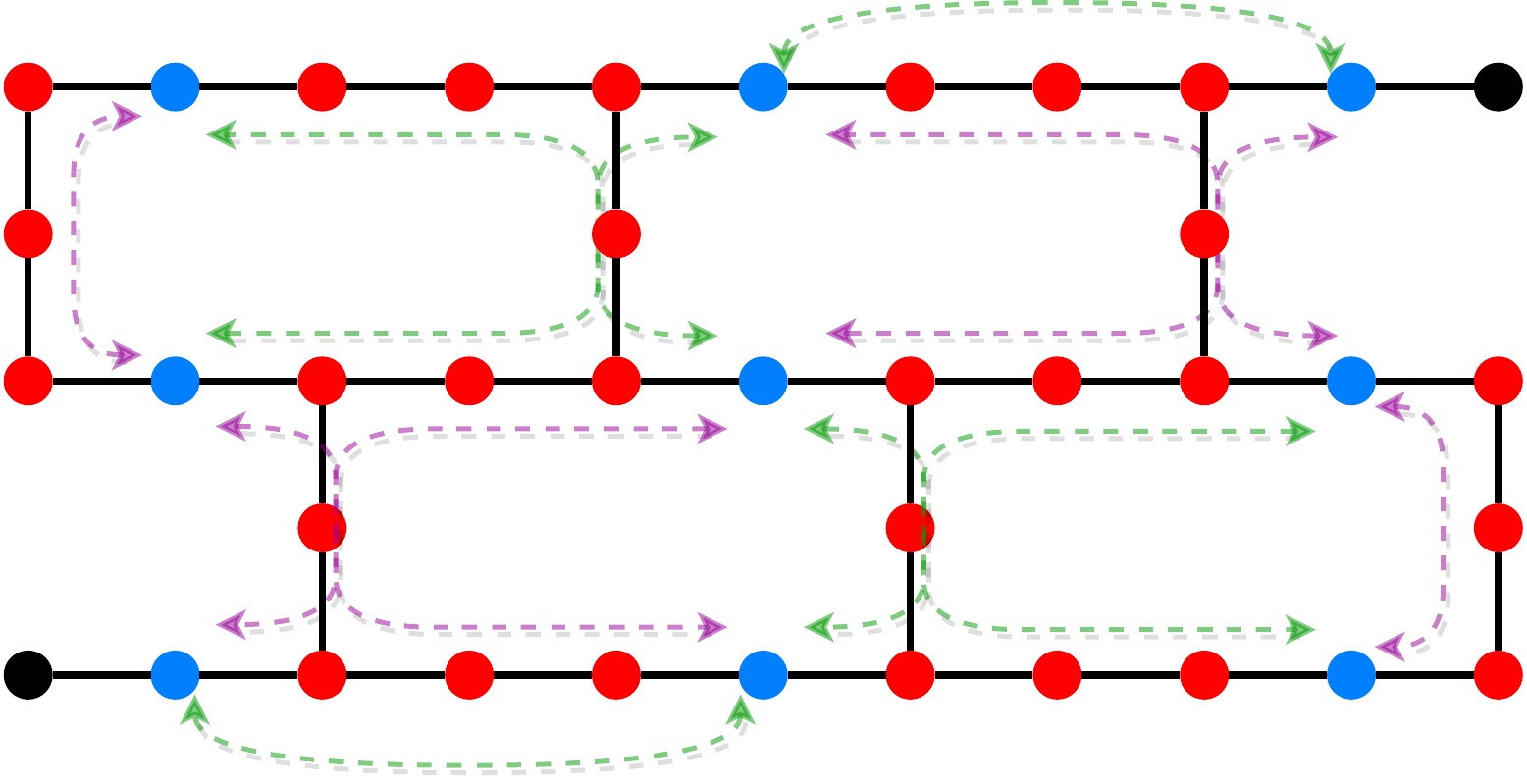}};
\node [below=0.01cm of image1]  {(b) \myCompilerNameSpace Heavy Hexagon};

\node [below=0.7cm of image, below right, inner sep=0] (image3) at (0,0) {\includegraphics[height=0.155\textwidth]{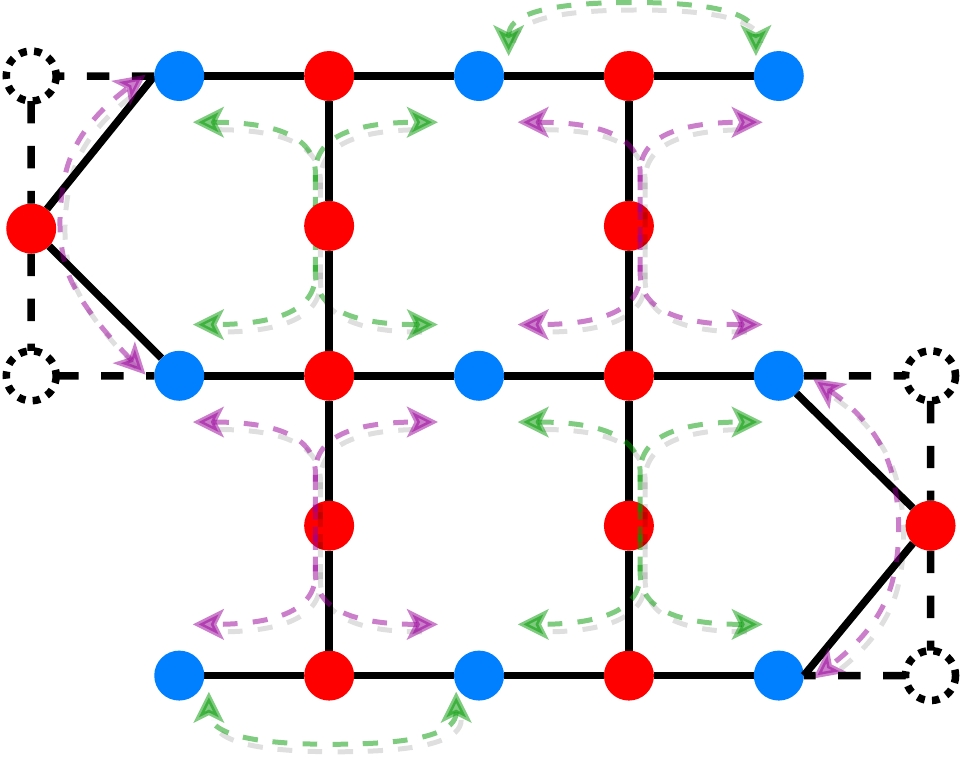}};
\node [below=0.01cm of image3]  {(c) IBM Heavy Square};

\node [below=0.7cm of image1, below right, inner sep=0] (image4) at (4.5,0) {\includegraphics[height=0.155\textwidth]{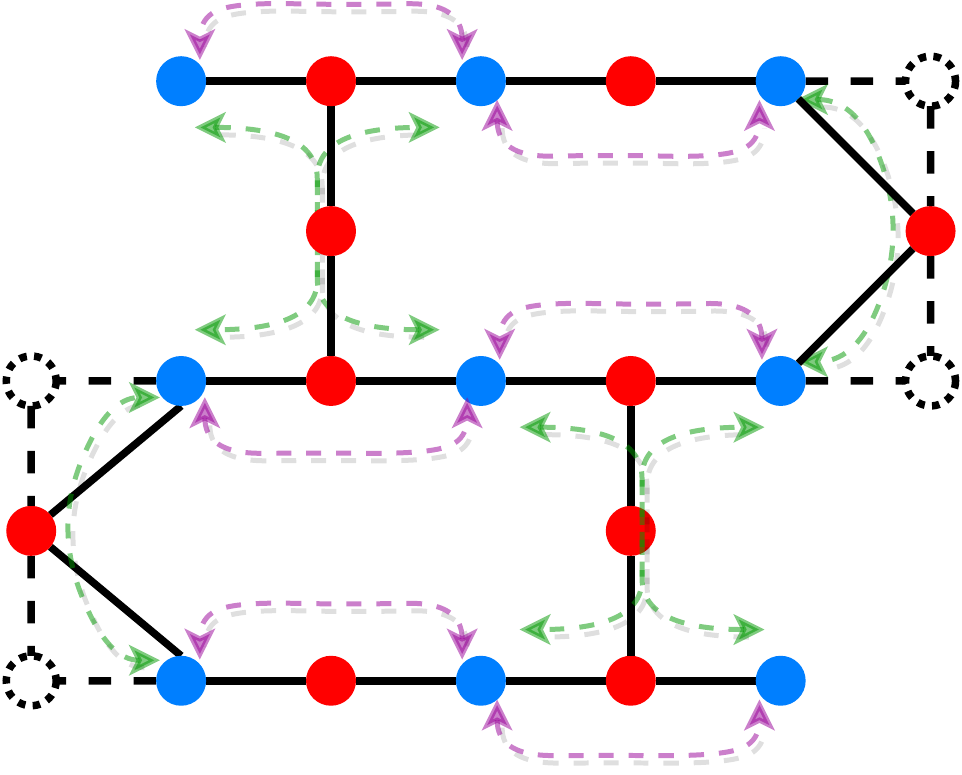}};
\node [below=0.01cm of image4]  {(d) IBM Heavy Hexagon};

\end{tikzpicture}
    \caption{The synthesized distance-3 surface code by \myCompilerNameSpace and the two manually designed QEC codes by IBM~\cite{Chamberland2020TopologicalAS}. IBM  removes some boundary nodes (dotted) and edges (dotted) for better efficiency of stabilizer measurements on borderline.  }
    \label{fig:our_ibm}
    
\end{figure}

\begin{figure}[htbp]
    \centering
        \begin{tikzpicture}
\node [above right,inner sep=0] (image) at (0,0) {\includegraphics[width=0.48\textwidth]{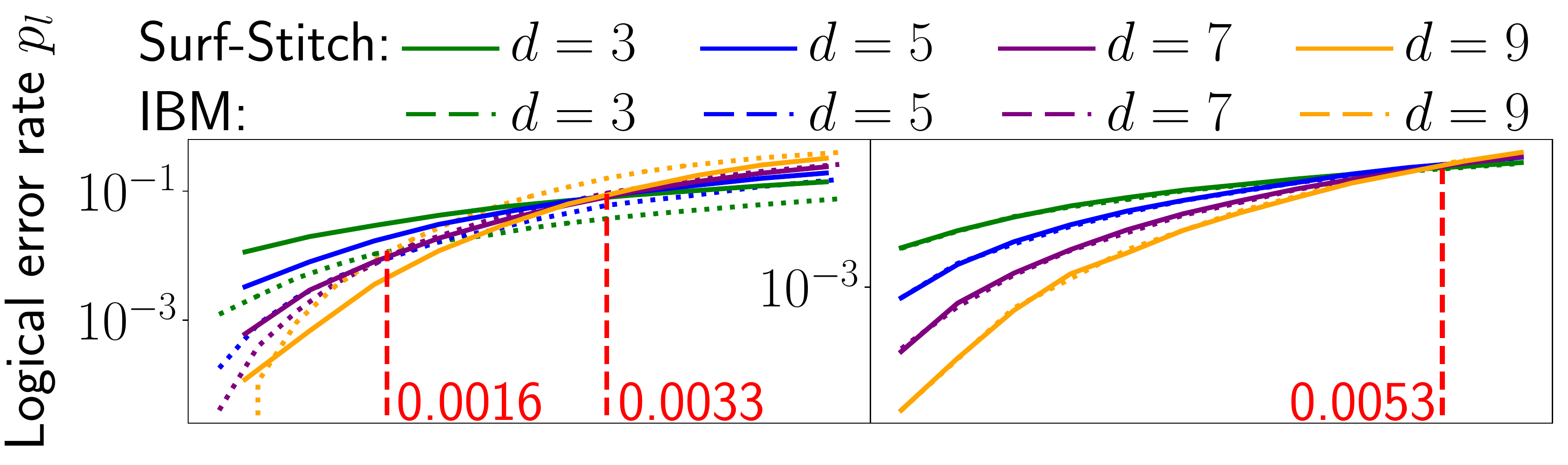}};
\node [below=-0.2cm of image]  at (2.9,0) {(a) Heavy Hexagon};
\node [below=-0.2cm of image]  at (6.9,0) {(b) Heavy Square};
\end{tikzpicture}
    \caption{The simulated error thresholds of the synthesized surface codes by \myCompilerNameSpace and the two manually designed QEC codes by IBM~\cite{Chamberland2020TopologicalAS}.
    }
    \label{fig:comp_ibm}
\end{figure}

\subsection{Synthesis on various SC architectures}

\begin{figure}
    \centering
    \begin{tikzpicture}
\node [above right,inner sep=0] (image) at (0,0) {\includegraphics[height=0.13\textwidth]{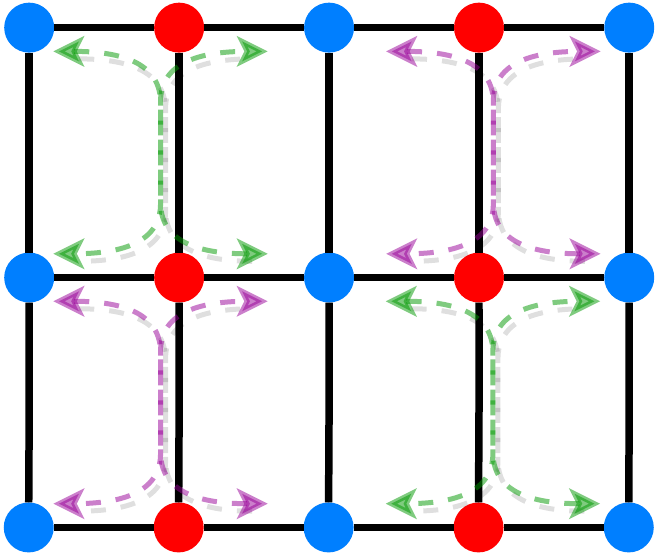}};
\node [below=0.01cm of image]  {(a) \myCompilerNameSpace Square};
\node [above right,inner sep=0] (image1) at (4,0) {\includegraphics[height=0.13\textwidth]{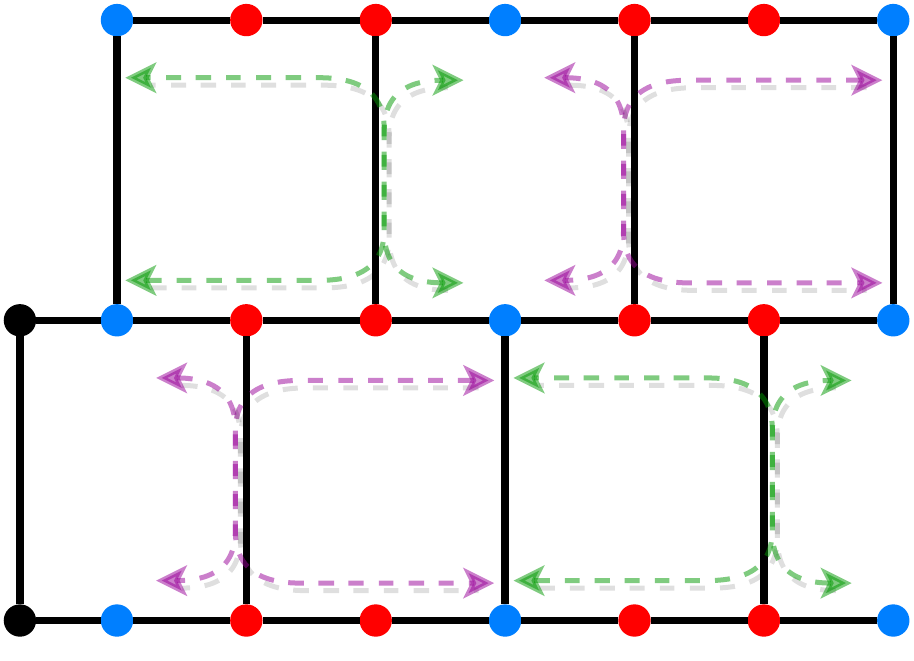}};
\node [below=0.01cm of image1]  {(b) \myCompilerNameSpace Hexagon};
\end{tikzpicture}
    \caption{First four stabilizers of the synthesized distance-5 surface code on square and hexagon architectures. }
    \label{fig:synthesis_poly}
\end{figure}

\begin{figure}
    \centering
    \begin{tikzpicture}

\node [above right, inner sep=0] (image2) at (0,0) {\includegraphics[height=0.16\textwidth]{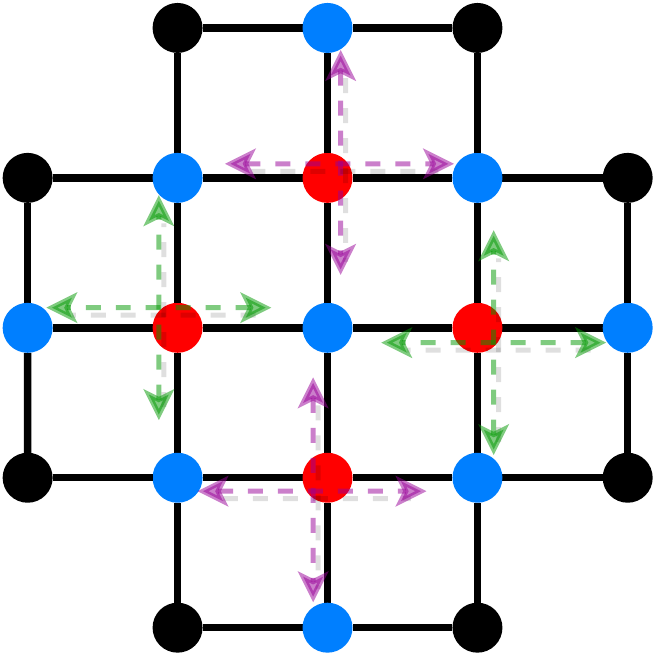}};
\node [below=0.01cm of image2]  {(a) \myCompilerNameSpace Square-4};
\node [above right, inner sep=0] (image2) at (4.5,0) {\includegraphics[height=0.16\textwidth]{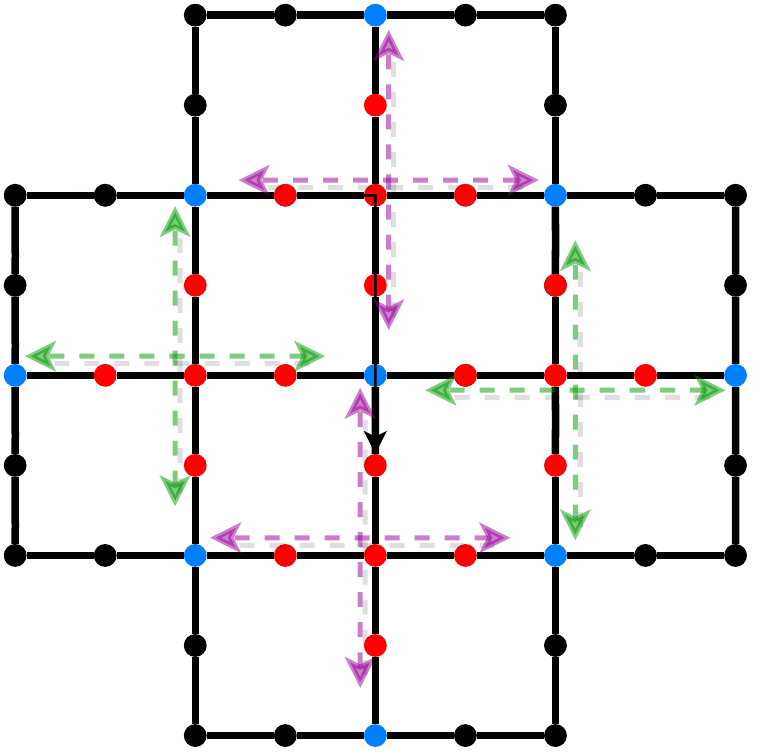}};
\node [below=0.01cm of image2]  {(b) \myCompilerNameSpace Heavy Square-4};
 
\end{tikzpicture}
    \caption{First four stabilizers of the synthesized distance-5 surface code by using syndrome rectangles induced by four-degree qubits. }
    \label{fig:synthesis_four}
\end{figure}

\begin{figure}[htbp]
    \centering
    \includegraphics[width=0.47\textwidth]{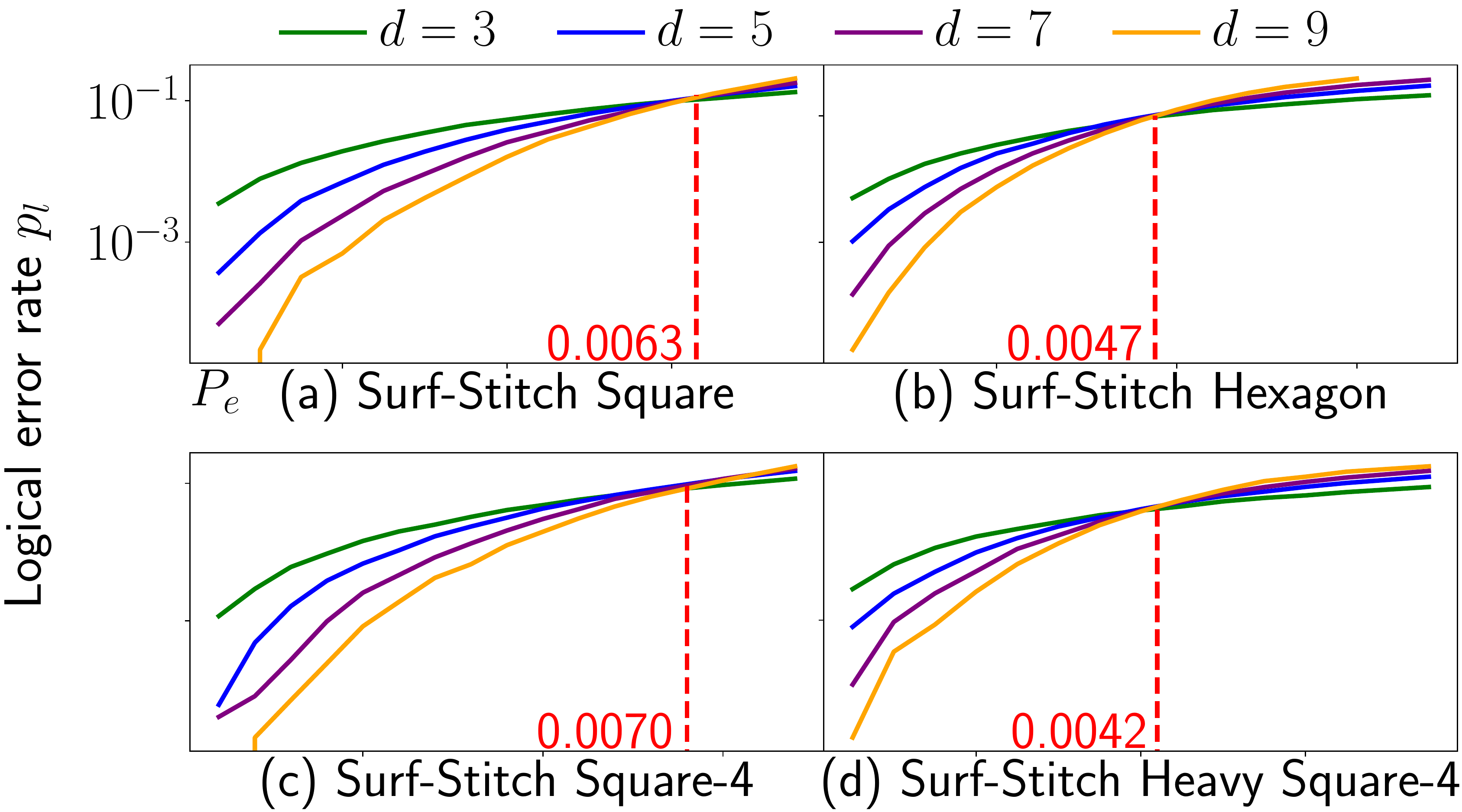}
    \caption{Numerical simulation results for the logical error rates of the synthesized surface codes on various SC architectures.
    }
    \label{fig:error_res}
\end{figure}

\begin{table}[h]
\caption{Error correction metrics of the synthesized surface codes. The average numbers of bridge qubits, CNOT gates, and time steps are computed over all X-type stabilizers.}
\label{tab:arch-result}
\resizebox{0.48\textwidth}{!}{%
\begin{tabular}{|m{1.9cm}|m{1.35cm}|m{1.2cm}|m{1.2cm}|m{1.2cm}|m{1.5cm}|}
\hline
Code  & Avg. bridge qubit \# & Avg. CNOT \# & Avg. time-step \# & Tot. time-step \# & Estimated error threshold \\ \hline
\myCompilerNameSpace Heavy Sqaure  & 3                   & 8           & 12  & 24 & 0.53\%         \\ \hline
\myCompilerNameSpace Heavy Hexagon & 7                   & 19          & 20 & 40 & 0.33\%           \\ \hline
\myCompilerNameSpace Sqaure        & 2                   & 6           & 10 & 20 &  0.63\%         \\ \hline
\myCompilerNameSpace Hexagon       & 4                   & 10          & 13 & 26 & 0.47\%        \\ \hline
\myCompilerNameSpace Sqaure-4       & 1                   & 4           & 8 & 8 & 0.70\%   \\\hline  
\myCompilerNameSpace Heavy Sqaure-4 & 5                   & 12           & 13  & 13 & 0.42\%  \\\hline
\end{tabular}%
}
\end{table}

\begin{table}[h]
\caption{Qubit utilization of the distance-5 surface codes synthesized by \myCompilerNameSpace on different architectures.}
\label{tab:qubit_utilize}
\resizebox{0.49\textwidth}{!}{%
\begin{tabular}{|m{1.9cm}|m{1.55cm}|m{1.75cm}|m{1.75cm}|m{1.5cm}|}
\hline \footnotesize
Code   &  data qubit \% & bridge qubit \% & unused qubit \% & Tot. qubit \# \\ \hline
\myCompilerNameSpace Heavy Sqaure   & 31.7\%         & 45.6\%           & 22.8\%   & 79       \\ \hline
\myCompilerNameSpace Heavy Hexagon & 18.8\%         & 59.4\%           & 21.8\%   & 133      \\ \hline
\myCompilerNameSpace Square         & 55.6\%         & 44.4\%           & 0.0\%   &  45        \\ \hline
\myCompilerNameSpace Hexagon        & 30.5\%         & 48.8\%           & 20.7\%    & 82      \\ \hline
\myCompilerNameSpace Square-4         & 43.9\%         & 56.1\%           & 0.0\%   &  57    \\ \hline
\myCompilerNameSpace Heavy Sqaure-4   & 16.3\%         & 83.7\%           & 0.0\%   & 153       \\ \hline
\end{tabular}%
}
\end{table}

We then apply \myCompilerNameSpace on the square architecture and the hexagon architecture to demonstrate that \myCompilerNameSpace can accommodate various architectures.
The synthesized surface codes on these two architectures are shown in Figure~\ref{fig:synthesis_poly}.
We also include another two surface code synthesis generated by using syndrome rectangles centering around four-degree qubits, as shown in Figure~\ref{fig:synthesis_four}.
Table~\ref{tab:arch-result} and Figure~\ref{fig:error_res} summarizes the error correction performance of these synthesized surface codes.
Table~\ref{tab:qubit_utilize} shows the resource requirement of the synthesized surface codes, and is obtained by finding the smallest tiling of building blocks for each architecture that is able to support the distance-5 surface code, and then computing the ratios of different types of qubits. 

\paragraph{The effect of architecture} High-degree architectures are more effective for surface code synthesis than low-degree architectures, both in error correction performance and resource requirement. Comparing to polygon architectures, heavy architectures reduces the error threshold by 26.7\% on average and they increases the average time-step number by 40.7\% averagely. Heavy architectures also increase bridge qubit number by 114\% averagely, up to 400\%. However, low-degree devices has a much lower physical error rate and are easier to fabricate than high-degree devices.

\paragraph{The effect of synthesis design} Synthesis centering four-degree qubits has higher resource overhead than the synthesis induced by a pair of three-degree qubits. In Table~\ref{tab:qubit_utilize}, 26.7\% and 93.7\%  more qubits are required for `\myCompilerNameSpace Square-4' and `\myCompilerNameSpace Heavy Square-4' than `\myCompilerNameSpace Square' and `\myCompilerNameSpace Heavy Square', respectively. 
Also, on low-degree architectures, the synthesis induced by four-degree qubits may have lower error threshold. Comparing to `\myCompilerNameSpace Heavy Square', `\myCompilerNameSpace Heavy Square-4' downgrades the error threshold by 20.8\%.

In summary, not only the architecture design but also the synthesis design have a critical impact on the resource overhead and error correction performance of the synthesized code.
By optimizing the three key steps in surface code synthesis, \textbf{our framework is able to mitigate the detrimental effect of low-degree architectures}. In fact, 
both using syndrome rectangles induced by four-degree qubits, `\myCompilerNameSpace Heavy Square-4' only achieves 60\% error correction performance of `\myCompilerNameSpace Square-4' (the ideal surface code) while `\myCompilerNameSpace Heavy Square' achieves 75.7\% error correction performance of `\myCompilerNameSpace Square-4'. On the other hand, comparing to the synthesized code on the hexagon architecture whose asymptotic average node degree is 3, `\myCompilerNameSpace Heavy Square' on the heavy square architecture, whose asymptotic average node degree is 2.3, achieves a 12.8\% higher error threshold. These evaluation results demonstrate the effectiveness of the proposed framework for automatic surface code synthesis on any SC architectures.

\subsection{Architecture Design Implications}

Our synthesis results provide some insights for designing future SC quantum architectures that can efficiently execute the surface code. Starting from the ideal 2D qubit array equipped with the surface code in Figure~\ref{fig:arch_implict}(a),
we discuss how to reduce the connectivity of the ideal 2D qubit array while preserving the efficiency for surface code synthesis.

\begin{figure}
    \centering
    \begin{tikzpicture}
\node [above right,inner sep=0] (image) at (0,0) {\includegraphics[height=75px]{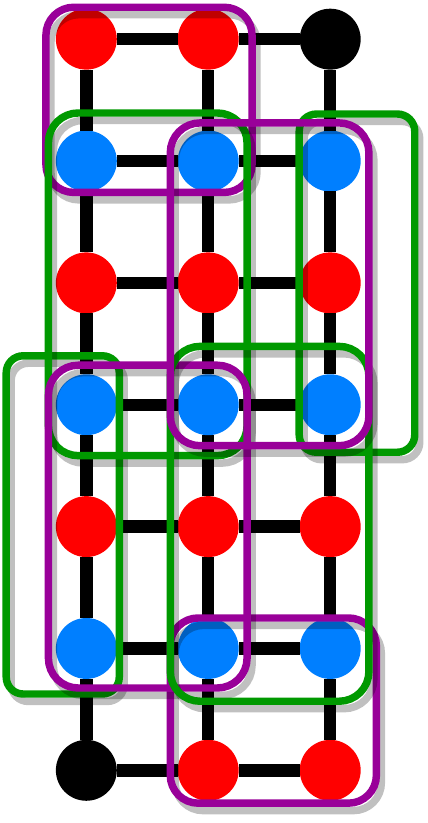}};
\node [below=0.01cm of image]  {(a)};

\node [above right,inner sep=0] (image1) at (2.2,0) {\includegraphics[height=75px]{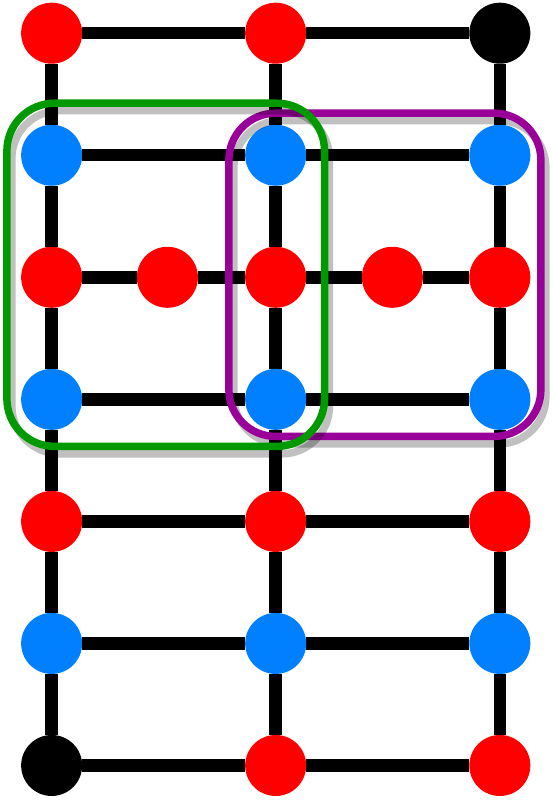}};
\node [below=0.01cm of image1]  {(b)};

\node [above right,inner sep=0] (image2) at (5,0) {\includegraphics[height=75px]{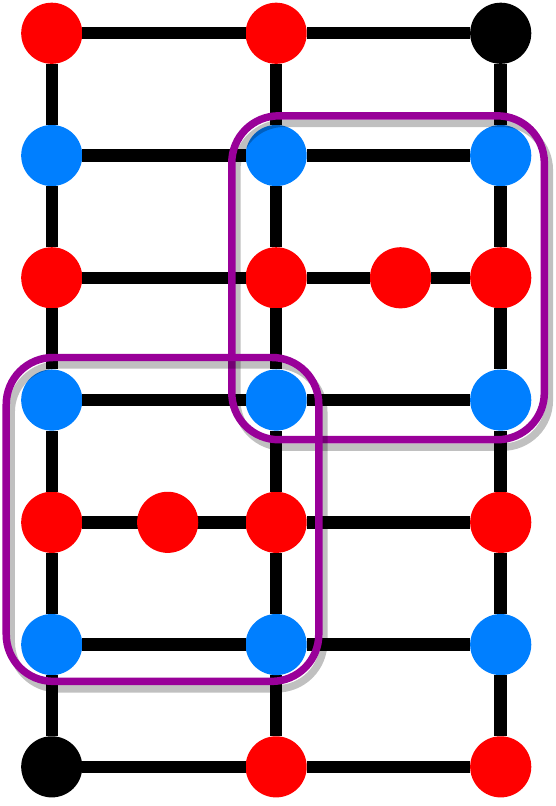}};
\node [below=0.01cm of image2]  {(c)};

\node [below=0.7cm of image, below right, inner sep=0] (image3) at (0,0) {\includegraphics[height=75px]{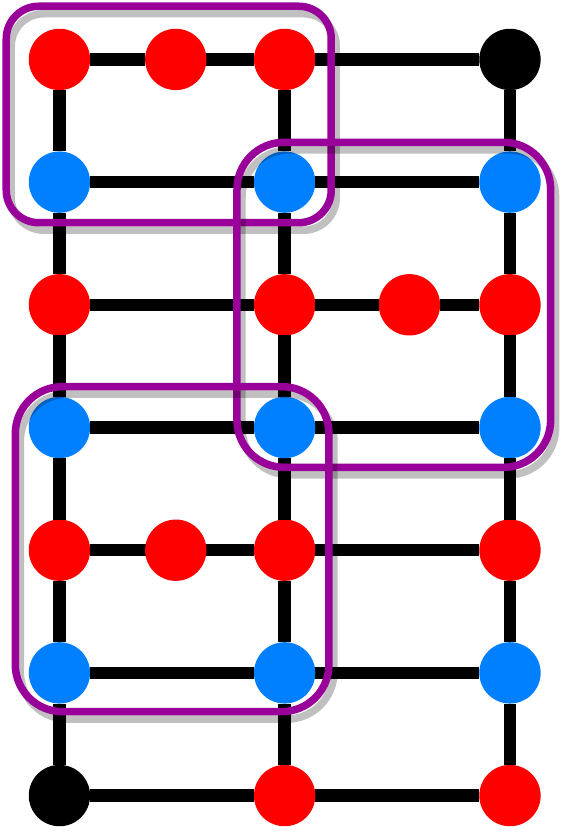}};
\node [below=0.02cm of image3]  {(d)};

\node [below=0.7cm of image1, below right,inner sep=0] (image4) at (2.5,0) {\includegraphics[height=75px]{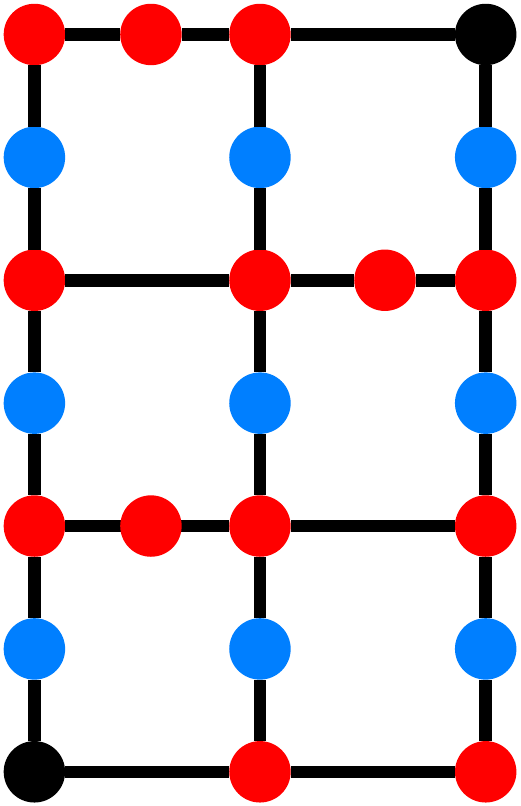}};
\node [below=0.01cm of image4]  {(e)};

\node [below=0.7cm of image2, below right,inner sep=0] (image5) at (5,0) {\includegraphics[height=75px]{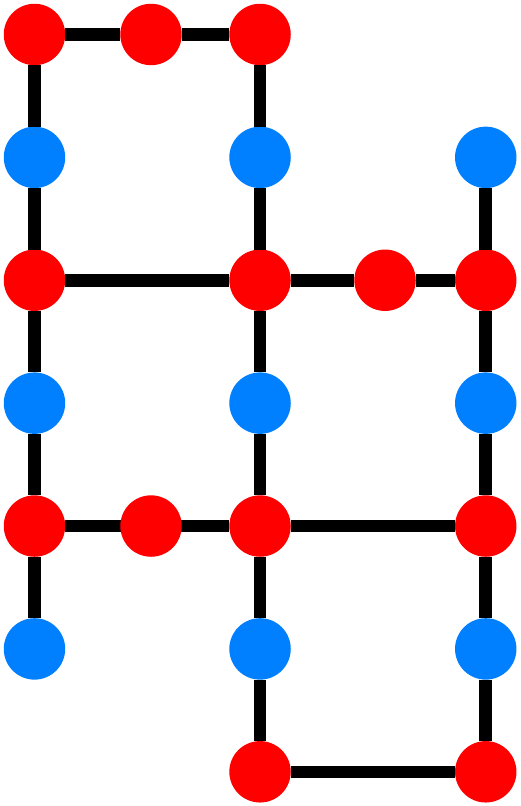}};
\node [below=0.01cm of image5]  {(f)};
\end{tikzpicture}
    \caption{An example for potential architecture design. }
    \label{fig:arch_implict}
\end{figure}
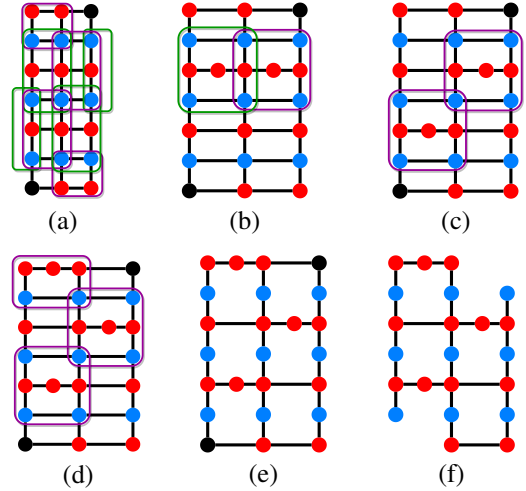

\paragraph{Insert two-degree nodes selectively}
Heavy edges or extra two-degree nodes are harmful to the synthesized surface code. Yet, they can be essential when lowering the device design and fabrication complexity~\cite{Li2020TowardsES}.
As a compromise, we can insert two-degree nodes selectively.

First, we should avoid inserting two-degree nodes both in X-type syndrome rectangles and Z-type syndrome rectangles. This is because we can confine the detrimental effect of two-degree nodes only within one stabilizer set in this way.  Figure~\ref{fig:arch_implict}(b) gives an example where two-degree nodes are inserted into both X- and Z-type syndrome rectangles. Comparing to the architecture in Figure~\ref{fig:arch_implict}(c) which only inserts nodes in Z-type syndrome rectangles, the architecture in Figure~\ref{fig:arch_implict}(b) will have two more time-steps in each error detection cycle. 
Second, it may be of merit to insert two-degree nodes in small syndrome rectangles of a scheduled stabilizer set.
A stabilizer set may have syndrome rectangles of different sizes. Inserting two-degree nodes in small syndrome rectangles may not increase the error detection cycle time. For example, the error detection cycle length of the architecture in Figure~\ref{fig:arch_implict}(d) is the same as the one in Figure~\ref{fig:arch_implict}(c).

\paragraph{Remove useless structures}
The unused qubits and physical connections for hexagon and heavy architectures are not a few, as indicated by Table~\ref{tab:qubit_utilize}. 
There are several types of useless structures in these architectures, including the architecture in Figure~\ref{fig:arch_implict}(d). {First}, connections between data qubits do not provide any benefits for constructing non-conflicted bridge trees. We can remove such connections to lower device complexity, as shown in Figure~\ref{fig:arch_implict}(e).
{Second}, 
qubits and connections on the boundaries of these architectures may not be used by any stabilizers and can thus be removed. For example, the qubits on the top right and bottom left corners of Figure~\ref{fig:arch_implict}(e) can be removed and the resulted architecture is shown in Figure~\ref{fig:arch_implict}(f).

\section{Related Work}

\textbf{Circuit compilation over the surface code:}
Most circuit compilation work on surface code are at the higher logical circuit level.
Javadi et al.~\cite{JavadiAbhari2017OptimizedSC} and Hua et al.~\cite{Hua2021AutoBraidAF} studied the routing congestion in circuit compilation over the surface code.
Ding et al.~\cite{Ding2018MagicStateFU} and Paler et al.~\cite{Paler2019SurfBraidAC} studied the compilation of magic state distillation circuits with existing surface code logical operations. %
These works %
usually assume that the ideal surface code array is already available and do not consider the problem of surface code synthesis on hardware.
In contrast, this paper focuses on optimizing the lower-level surface code synthesis on various SC architectures.

\textbf{QEC code and architecture:}
most efforts on QEC code synthesis are still on looking for an architecture that is suitable for the target code.
Reichardt~\cite{Reichardt2018FaulttolerantQE} proposes three possible planar qubit layouts for synthesizing the seven-qubit color code.
Chamberland et al.~\cite{Chamberland2019TriangularCC} proposes a trivalent architecture where it is straightforward to allocate data qubits of triangular color codes. 
Chamberland et al.~\cite{Chamberland2020TopologicalAS} introduces heavy architectures which reduces frequency collision while still provides support for surface code synthesis. Instead, the synthesis framework in this paper can automatically synthesize the surface code onto various mainstream architectures and avoid manually redesigning code protocols for the ever-changing architectures.
Another line of research targets at compiling stabilizer measurement circuits to existing architectures.
Lao and Almudever \cite{Lao2020FaulttolerantQE} proposes the flag-bridge circuit which can measure the stabilizer of the Steane code on the IBM-20 device. However, their work relies on manually appointed data qubits and bridge qubits, and focuses on the IBM-20 device. Methods in this category are orthogonal to our work, and can be easily merged into our framework.

\section{Discussion}

Though we propose a comprehensive framework to synthesize surface code on various SC devices, there is still much space left for potential improvements. 
For example, besides heuristic data qubit allocation schemes, it is also promising to train a neural network to allocate data qubits.
Non-local bridge trees may also be explored for better parallelism of stabilizer measurements. 
It is also possible to merge bridge trees to resolve bridge qubit conflict. However, this requires careful tuning since large bridge trees may be detrimental for error detection.
We do not include the error decoder design as our work mainly focus on surface code synthesis. Though we can reuse previous surface code decoders~\cite{Chamberland2020TopologicalAS, Varsamopoulos2017DecodingSS, Krastanov2017DeepNN, Baireuther2017MachinelearningassistedCO}, there may be some opportunities to devise more accurate error decoders for the proposed surface code synthesis scheme.

Another interesting future direction is to adapt our synthesis framework to other error correction codes. Though surface code can be implemented on various SC devices, it may not be the most efficient one. Extending our framework to other error correction codes can help us fully exploiting existing device architectures for FT computation.
On other hand, the proposed surface code synthesis framework can provide guidance for SC architecture design. Using our framework as the baseline, we can compare the efficiency of different device architectures for implementing surface code and identify the most efficient one.

\section{Conclusion}
In this paper, we formally describe the three challenges of synthesizing surface code on SC devices and present a comprehensive synthesis framework to overcome these challenges. The proposed framework consists of three optimizations.
First, we adopt a geometrical method to allocate data qubits in a way that ensures the existence of shallow measurement circuits. 
Second, we only consider bridge qubits enclosed by data qubits and reduce the number of bridge qubits by merging short paths between data qubits. The proposed bridge qubit optimization reduces the resource conflicts between syndrome measurement. 
Third, we propose an iterative heuristic to schedule the execution of measurement circuits based on the proposed data qubit allocation. %
Our comparative evaluation to manually designed QEC codes demonstrates that, with good optimization, automated synthesis can surpass manual QEC code design by experienced theorists.

\bibliographystyle{unsrt}
\bibliography{references}

\end{document}